\documentclass{amsart}

\usepackage{amsmath,amssymb,amsthm}
\usepackage{graphicx}
\usepackage{dcolumn}% Align table columns on decimal point
\usepackage{multirow}
\usepackage{float}
\usepackage{cite}
\usepackage{color}
\usepackage{url}

\newcommand{\RR}{\mathbb{R}}
\newcommand{\CC}{\mathbb{C}}
\newcommand{\NN}{\mathbb{N}}
\newcommand{\ud}{\textup{d}}
\newcommand{\eMP}{\epsilon\textup{-MP}}
\newcommand{\MP}{\textup{MP}}
\DeclareMathOperator*{\argmin}{arg\,min}
\theoremstyle{definition}
\newtheorem{theorem}{Theorem}[section]
\newtheorem{defn}[theorem]{Definition}

\newtheorem{thm}{Theorem}[section]

\begin{document}

\title[Minimum-latency Time-frequency Analysis]{Minimum-latency Time-frequency Analysis Using Asymmetric Window Functions}

\author{Li~Su~ and~Hau-tieng~Wu}
\address{L. Su is with Research Center for Information Technology Innovation, Academia Sinica, Taipei, Taiwan, e-mail: (lisu@citi.sinica.edu.tw)}
\address{H.-T. Wu is with the Department of Mathematics, University of Toronto, Toronto ON M5S 2E4, Canada. e-mail: (hauwu@math.toronto.edu)}

\maketitle

\begin{abstract}
We study the real-time dynamics retrieval from a time series via the time-frequency (TF) analysis with the minimal latency guarantee.
While different from the well-known intrinsic latency definition in the filter design, a rigorous definition of intrinsic latency for different time-frequency representations (TFR) is provided, including the short time Fourier transform (STFT), synchrosqeezing transform (SST) and reassignment method (RM).
To achieve the minimal latency, a systematic method is proposed to construct an asymmetric window from a well-designed symmetric one based on the concept of minimum-phase, if the window satisfies some weak conditions.
We theoretically show that the TFR determined by SST with the constructed asymmetric window does have a smaller intrinsic latency.
Finally, the music onset detection problem is studied to show the strength of the proposed algorithm.
\end{abstract}

\section{Introduction}

Oscillation is ubiquitous. The dynamics of the oscillatory behavior usually serve as a portal to the system. When the system is non-linear and non-stationary, the traditional Fourier transform based on the assumption of stationary input is unsuitable, and time-frequency (TF) analysis becomes a natural candidate for this purpose \cite{Flandrin:1999}.
There are several TF analysis methods available, ranging from linear, quadratic to nonlinear ones \cite{Flandrin:1999}. For most of these methods, the instantaneous spectrum at a time instance $t$ is estimated based on a windowed segment of the signal around $t$, and the estimated spectra over time constitute the time-frequency representation (TFR). The process of signal segmentation along with the computation of the spectrum introduces {\em latency}, that is, delay of the output signal in comparison to the input signal. Take the short-time Fourier transform (STFT) as an example. For a {\em symmetric} window function $h(t)\in L^2$ supported on $[-T/2, T/2]$, the STFT of $x\in L^2(\mathbb{R})$ is represented as
\begin{equation}
V^{(h)}_x(t, \omega) = \int x(s) h(s-t)e^{-j2\pi\omega s} ds,
\label{eq: stft1}
\end{equation}
where $t\in\RR$ is the time and $\omega\geq 0$ is the frequency. We call $|V^{(h)}_x|^2$ the {\em spectrogram} of the signal and $|V^{(h)}_x(t, \cdot)|^2$ the {\em spectrum} of the signal at $t$. By the symmetry assumption, $h(t)$ is centered at $t$, and therefore we estimate the spectrum at $t$ based on the windowed segment in $[t-T/2, t+T/2]$. In other words, we obtain the spectrum estimated at $t$ only after we collect data up to time $t+T/2$. Therefore, there is {\em at least} a latency $T/2$ presented in the STFT when compared with the input signal, and the total latency should be always larger then $T/2$, no matter how the hardware or algorithm are implemented. We define this required time as the \emph{intrinsic latency}.

Nowadays, extracting dynamics using TF analysis with real-time speed and low latency has emerged as a fundamental requirement in various interactive and user-centered technologies, ranging from hearing aids \cite{mauler2010optimization}, computer-aided music practicing tools, brain-computer interface \cite{lin2010real}, to patient monitoring systems. For example, in the patient monitor system, it is allowed to display the minute heart rate with the maximal 10-s delay \cite{AMS:2007} for the clinical monitoring purpose, and similar criteria are needed for the information acquired from the TF analysis. A musical pitch tracking system needs a 10 ms latency while a commonly used window length is 93 ms, a challenge exemplified in \cite{pardue2014lowlatency}: the system specification is never satisfied as the latency is at least $93/2=46.5$ ms, no matter how fast the algorithm can be. In these cases, the intrinsic latency becomes the bottleneck for real-time processing. One intuitive way to reduce the intrinsic latency is shortening the window, but doing so sacrifices the spectral resolution at the same time due to the Heisenberg uncertainty principle \cite{Flandrin:1999}.

This issue could be alleviated by replacing the symmetric window $h(t)$ with an {\em asymmetric} window to reduce the intrinsic latency \cite{florencio1991use, alam2014robust, morales2012use, alam2012robust, kotnik2002robust, Rozman_Kodek:2007, heo2015classification, mauler2007low, nagathil2012optimal, lollmann2008low, mauler2010optimization, mauler2009improved}. In a nutshell, an asymmetric window $h^\circ(t)$ with the support $[-T/2, T/2]$ with the energy concentrated at $t^\circ>0$ could make the intrinsic latency less than $T/2$, since the ``window mass'' is toward the future. Recently, such an idea has also been discussed in the context of adaptive TF analysis \cite{Andersen2016}. However, to the best of our knowledge, this issue in the context of nonlinear TF analysis, like the synchrosqueezing transform (SST) and reassignment method (RM), is less discussed. Further, the discussion of the interaction between the asymmetric window and the TF analysis is limited, possibly due to the lack of a clear definition of latency in the context of TF analysis.

\subsection{Defining latency for TF analysis: an extension from FIR filter design}

The most well-known definition of the intrinsic latency could be seen in filter design. Similar to the window function in TF analysis, the filter introduces intrinsic latency because of its finite length. Given an input signal $x(t)$ and a finite-impulse-response (FIR) filter $h(t)$, the intrinsic latency of a filter is known as the {\em group delay} of its transfer function. Precisely, the filtered (output) signal of $x$ is represented as $y(t)=h(t)*x(t)$, or $\hat{y}(\omega)=\hat{h}(\omega)\hat{x}(\omega)$, where $\hat{\cdot}$ is the Fourier transform. Let $\hat{h}(\omega)=|\hat{h}(\omega)|e^{i\Phi^{(h)}(\omega)}$. The group delay is represented as the negative time derivative of the phase function of $\hat{h}(\omega)$:
\begin{equation}
\Gamma^{(h)}(\omega):=-\frac{\partial \Phi^{(h)}(\omega)}{\partial\omega}=\mathfrak{R} \frac{\widehat{\mathcal{T}h}(\omega)}{\hat{h}(\omega)}
\label{eq:gdf_filter}
\end{equation}
where $(\mathcal{T}h)(t)=t h(t)$. Designing a {\em minimum-phase (MP) filter} is a well-known technique for minimizing the intrinsic latency of a filter \cite{Oppenheim_Schafer:2009}. A MP filter is a MP sequence, which has the smallest group delay among all transfer functions having the same magnitude in the frequency domain.

While this notion of the filter's intrinsic latency is widely applied in the signal processing field, however, it cannot be directly applied to TF analysis, since the definition of group delay in TF analysis is fundamentally different from that in the FIR filter. First, the group delay of the FIR filter is fully determined by the filter transfer function, but the one of a TFR depends on not only the window function, but also the input signal and time, due to non-linear and non-stationary features of the signal. This could be seen by comparing (\ref{eq:gdf_filter}) and (\ref{eq:gdf}). Second, in FIR filters, we are only interested in the signal lying in the {\em pass-band} of the filter, but in TF analysis we need to consider all frequencies since the structure of the signal is unknown.

It is also not all that direct to apply the MP filter design method to design a window minimizing the latency in TF analysis. For example, most of the commonly-used symmetric windows, like the Hamming or Hann windows, cannot be converted into asymmetric MP windows because all of their zeros in the z-plane lie on the unit circle. This implies that, in addition to properties in the spectrum like main-lobe width, side-lobe level and others, the position of zeros in the z-plane should be a condition in designing MP windows.

In brief, to relate the notion of latency in filters to TF analysis seems intuitive, but by no means trivial. In the following sections, we will introduce more theoretical backgrounds and techniques to connect the two fields. From the discussion in Section \ref{sec: tf_latency}, we will show that the intrinsic latency of a filter could be successfully extended to TF analysis. More specifically, in our discussion, an FIR low-pass filter (LPF) can be regarded as a special case of a window function at zero frequency in the TFR.

\subsection{Our contribution}

In this study, we give a systematic discussion on the latency minimization of three intimately related but different TF analysis techniques, including the STFT, SST and RM, by using asymmetric windows.
First, we give a systematic discussion about how the notion of intrinsic latency depends on the TF analysis and the signal.
Second, we provide a signal-independent definition of intrinsic latency for an asymmetric window, which serves as a base for the discussion of latency.
Third, to balance the trade-offs between latency minimization and the spectral deformation provided, the defined intrinsic latency is applied to construct asymmetric windows with theoretical guarantee for the TF analysis.
To demonstrate the result, we apply the proposed method on the real-time musical onset detection task. Results show the benefit of using an asymmetric window for achieving a very low latency.
To the best of our knowledge, a rigorous definition of intrinsic latency of an asymmetric window, the influence of the asymmetric window on the TF analysis, as well as the design of asymmetric windows for the TF analysis, were all unexplored in the past.

This paper is organized as follows. Section \ref{sec: tf_latency} summarizes TF analysis techniques, including STFT, SST and RM, and the notion of intrinsic latency is provided and justified. Section \ref{sec:mlwin} introduces a systematic way of constructing asymmetric windows with the minimal intrinsic latency for TF analysis techniques.
Section \ref{sec:application} goes for the application on real-time musical onset detection and results.

\section{TF representation considering latency}
\label{sec: tf_latency}

In this section, we introduce TF analysis methods discussed in this paper. To simplify the discussion, we focus on typical linear and nonlinear methods including STFT, SST and RM. The Wigner-Ville distribution (WVD) is also introduced for the need of a detailed analysis of intrinsic latency. Then we provide a definition of intrinsic latency inherited from the window, and discuss how the TFR are influenced by the intrinsic latency caused by the chosen window.

\subsection{Adaptive harmonic model}

Fix a sufficiently small non-negative constant $\epsilon$, $d>0$, and two positive constants $c_2\geq c_1$ larger than $\epsilon$. Consider the set $\mathcal{A}_{\epsilon,d}^{c_1,c_2}$ consisting of functions $x(t) = \sum_{l=1}^LA_l(t)\cos(2\pi\phi_l(t))$, where $L\in\NN$ and for each $l=1,\ldots,L$ and for all $t\in\RR$, the function satisfies the
{\em regularity conditions}:
\begin{equation}
A_l\in C^1(\RR)\cap L^\infty(\RR),\,\, \phi_l\in C^2(\RR),\label{AHM:condition1}
\end{equation}
the {\em boundedness conditions}:
\begin{align}
c_1\leq  A_l(t)\leq c_2,\,\, c_1\leq \phi_l'(t)\leq c_2\label{AHM:condition2}
\end{align}
and the {\em slowly varying} conditions
\begin{equation}
|A_l'(t)|\leq \epsilon\phi'(t),\quad |\phi_l''(t)|\leq \epsilon\phi'(t)\label{AHM:condition3}.
\end{equation}
Further, if $L>1$, we have the {\em separation condition}
\begin{equation}\label{definition:adaptiveHarmonicMultiple}
\phi_{\ell+1}'(t)-\phi'_\ell(t)>d
\end{equation}
for all $\ell=1,\ldots,L-1$ and $t\in\RR$. We call $A_l(t)\cos(2\pi\phi_l(t))$ in $x(t)$ an {\em intrinsic mode type (IMT) function}, and $\mathcal{A}_{\epsilon,d}^{c_1,c_2}$ the {\em adaptive harmonic model}.
Intuitively, we could call $A_l(t)$ the {\em amplitude modulation (AM)} and $\phi'_l(t)$ the {\em instantaneous frequency (IF)} of the $l$-th IMT function, and these nominations are justified by the identifiability issue discussed in \cite{Chen_Cheng_Wu:2014}.

\subsection{Time-frequency representations}

\subsubsection{STFT} For a given function $x$ and window $h$ in the proper space, STFT is defined in (\ref{eq: stft1})  \cite{Flandrin:1999}. In this paper, we assume that $x$ is a tempered distribution and $h$ is a Schwartz function. Note that functions satisfying the adaptive harmonic model are tempered distributions.

\subsubsection{SST} The main limitation of STFT is the ``spreading'' issue caused by the Heisenberg uncertainty principle. To sharpen the TFR determined by STFT, or the spectrogram, the reassignment technique was first introduced in \cite{Kodera_Gendrin_Villedary:1978}, and then improved by \cite{Auger_Flandrin:1995}. We refer the reader with interest to \cite{Daubechies_Wang_Wu:2016} for a summary of the current progress in this direction. SST is a special reassignment technique, which counts on the {\em frequency reassignment rule} to sharpen the TFR, which is defined in every points $(t,\eta)$ by:
\begin{align}
\Omega^{(h,\Theta)}_x(t,\eta)=-\Im\frac{V_x^{(\mathcal{D}h)}(t,\eta)}{V_x^{(h)}(t,\eta)}\mbox{ when }  |V_x^{(h)}(t,\eta)|>\Theta\label{RM:omega}
\end{align}
and $\Omega^{(h,\Theta)}_x(t,\eta)=-\infty$ otherwise,
where $\Im$ means taking the imaginary part, $\Theta\geq0$ is the chosen hard threshold to reduce the numerical error and noise influence and $\mathcal{D}h:=h'$, the first derivative of $h$.
The main idea in SST is that the phase hidden in STFT includes the correct ``frequency information'' about the signal. Particularly, when the signal $x(t)$ satisfies the adaptive harmonic model, it has been shown in \cite{Wu:2011thesis} that
\begin{equation}\label{Expansion:OmegaFunctionOverAdaptive}
\Omega^{(h,\Theta)}_{x}(t,\eta)\approx \phi_l'(t)-\eta
\end{equation}
when $\eta\in [\phi'_l(t)-\Delta,\phi'_l(t)+\Delta]$, where the Fourier transform of $h$ is essentially supported on $[-\Delta,\Delta]$. Thus, $\Omega_x^{(h,\Theta)}$ does provide the IF information, which could be taken into account to sharpen the TFR by ``reassigning'' the coefficients of STFT, leading to the SST algorithm. In \cite{Daubechies_Lu_Wu:2011}, SST is defined as:
\begin{align}\label{definition:SST}
S^{(h,\Theta,\alpha)}_{x}(t,\xi):=\int_{\mathfrak{N}_t}V^{(h)}_x(t,\eta)g_\alpha\left(|\xi-\eta-\Omega^{(h,\Theta)}_x(t,\eta)|\right)\ud \eta,
\end{align}
where $t\in\RR$, $\xi\geq0$, $\mathfrak{N}_t:=\{\eta:\,|V^{(h)}_x(t,\eta)|>\Theta\}$, $0<\alpha\ll 1$ is chosen by the user, $g_{\alpha}(\cdot):=\frac{1}{\alpha}g(\frac{\cdot}{\alpha})$ and $g$ is a smooth function so that $g_\alpha\to \delta$ in the weak sense as $\alpha\to 0$. In other words, in SST, we nonlinearly move STFT coefficients {\em only} on the frequency axis.

\subsubsection{RM} Compared with SST, STFT coefficients are not just moved on the frequency axis, but also on the time axis in RM. Specifically, by defining the {\em STFT group delay} by
\begin{equation}
\Gamma^{(h,\Theta)}_{x}(t,\eta) :=
\mathfrak{R}\frac{V_x^{(\mathcal{T}h)}(t, \eta)}{V_x^{(h)}\left(t, \eta \right)}\mbox{ when }  |V_x^{(h)}(t,\eta)|>\Theta
\label{eq:gdf}
\end{equation}
and $\Gamma^{(h,\Theta)}(t,\eta)=-\infty$ when $|V_x^{(h)}(t,\eta)|\leq \Theta$, where $\mathfrak{R}$ means taking the real part and $\mathcal{T}h\left(t\right):=t \cdot h\left(t\right)$, the RM is defined as
\begin{align}
R^{(h,\Theta,\alpha)}_{x}(t,\xi):=&
\int_{\mathfrak{N}_t}\int_{\mathfrak{M}_\xi}V^{(h)}_x(s,\eta)g_\alpha\left(|\xi-\eta-\Omega^{(h,\Theta)}_x(t,\eta)|\right)\nonumber\\
&\times g_\alpha\left(|t-s-\Gamma^{(h,\Theta)}_x(s,\xi)|\right)\ud s\ud \eta,\label{definition:RM}
\end{align}
where $\mathfrak{M}_\xi:=\{s:\,|V^{(h)}_x(s,\xi)|>\Theta\}$. In comparison with (\ref{eq:gdf_filter}), (\ref{eq:gdf}) stems from the derivative of the phase of $V^{(h)}_x$ with respect to frequency \cite{Auger_Flandrin:1995}. Note that in \cite{Auger_Flandrin:1995}, the RM is defined by reassigning the spectrogram $|V^{(h)}_x|^2$ instead of $V^{(h)}_x$ as described in (\ref{definition:RM}). Here, to be consistent and emphasize the relationship between SST and RM, we stick our definition to (\ref{definition:RM}). We emphasize that the STFT group delay is different from the notion of {\em group delay of a window function}.

Throughout this paper, we set $\Theta$ as $10^{-4}\%$ of the root mean square energy of the signal under analysis and $\alpha$ small enough so that $g_\alpha$ is implemented as a discretization of the Dirac measure. When there is no danger of confusion, we will omit $\Theta,\alpha$ in the notation and simply denote $S^{(h,\Theta,\alpha)}_{x},R^{(h,\Theta,\alpha)}_{x},\Omega_{x}^{(h,\Theta)},\Gamma_{x}^{(h,\Theta)}$ as $S^{(h)}_{x},R^{(h)}_{x},\Omega_{x}^{(h)},\Gamma_{x}^{(h)}$.

\subsubsection{Wigner-Ville distribution} We need WVD \cite{Grochenig:2001} to analyze the intrinsic latency. Take $x\in L^2(\RR)$. The WVD of $f$, denoted as $W_x$, is defined as
\begin{equation}
W_x(\tau,\eta):=\int_{\RR} f(\tau+t/2)f^*(\tau-t/2)e^{-2\pi i\eta t}dt,
\end{equation}
where $\tau,\eta\in\RR$. More generally, when $x(\tau+\cdot/2)x^*(\tau-\cdot/2)$ is a tempered distribution for a fixed $\tau$, the WVD of $x$ could still be defined in the distribution sense. For example, we could define the WVD of a function $x$ satisfying the adaptive harmonic model since $x\in C^1(\RR)\cap L^\infty(\RR)$ and hence $x(\tau+\cdot/2)x^*(\tau-\cdot/2)$ is a tempered distribution.
The key ingredient we need from the WVD is its well-known relationship with the spectrogram for a suitable $x$ and $h$:
\begin{equation}
|V_x^{(h)}(t,\xi)|^2=\iint W_x(\tau,\eta)W_h(\tau-t,\eta-\xi)\ud \tau\ud \eta.\label{Equation:KeyRelationshipWVDSTFT}
\end{equation}

\subsection{The intrinsic latency for TF analysis}
\label{sec:latency}

We now discuss the intrinsic latency of STFT, SST and RM, respectively. Throughout this section, we fix a window $h$ supported on $[-T/2,T/2]$, where $T>0$. Fix a time $t$. In general, for the window $h$, we have to collect the signal up to time $t+T/2$ before we could evaluate Equation (\ref{eq: stft1}) at time $t$.
To emphasize this fact, call
\begin{equation}
t^{(h)}_o:=t+T/2
\end{equation}
the \emph{observation time} associated with the window $h$ at time $t$, which is the {\em front} of the needed signal for STFT at time $t$ associated with the window $h$. Note that the observation time is the same for all windows supported on $[-T/2,T/2]$.

Intuitively, intrinsic latency refers to the difference between the observation time and the time associated with the ``obtained information''. Ideally, the notion of intrinsic latency should depend {\em only} on the window function, but we show here that for all TFRs, this notion might also depend on the signal of interest. Since IF and AM are the most important information we want to obtain from a function satisfying the adaptive harmonic model, we focus on these features to study the intrinsic latency of TF analysis.
We start from stating the following theorem, which considers functions more general than the adaptive harmonic model \cite{Kowalski_Meynard_Wu:2015}.
\begin{thm}\label{Theorem:OmegaGamma}
Fix $\epsilon\geq 0$ sufficiently small, $c_2>c_1>\epsilon$ and $c_3\geq 0$. Take a function $x(t)=a(t)e^{i2\pi\phi(t)}$, where $a(t)$ and $\phi(t)$ satisfy (\ref{AHM:condition1}), (\ref{AHM:condition2})
and $|a'(t)|\leq \epsilon\phi'(t)$, $|\phi''(t)|\leq c_3$ and $|\phi'''(t)|\leq \epsilon\phi'(t)$ over $[t_0-M/2,t_0+M/2]$, where $t_0\in\RR$ and $M>0$. Take a smooth real window $h$ of unit norm and supported on $[-T/2,T/2]$ with $0<T<M$. Then at time $t_0$,
the frequency reassignment rule satisfies
\begin{align}
\Omega_{x}^{(h)}&(t_0,\xi)=\phi'(t_0)+\phi''(t_0)\frac{\int \tau W_h(\tau,\phi''(t_0)\tau+\phi'(t_0)-\xi)\ud \tau}{\int W_h(\tau, \phi''(t_0)\tau+\phi'(t_0)-\xi)\ud \tau}+O(\epsilon)\nonumber
\end{align}
and the STFT group delay satisfies
\begin{equation}
\Gamma_{x}^{(h)}(t_0,\xi) = \frac{\int \tau W_h(\tau,\phi''(t_0)\tau+\phi'(t_0)-\xi)\ud \tau}{\int W_h(\tau, \phi''(t_0)\tau+\phi'(t_0)-\xi)\ud \tau}+O(\epsilon)\nonumber\,.
\end{equation}
\end{thm}

We call the function $x(t)=a(t)e^{i2\pi\phi(t)}$ the \textit{generalized IMT} function and $\phi''(t)$ the \textit{chirp factor} of the function.

\begin{proof}
Without loss of generality, we could assume $a(t_0)=1$. By taking the Taylor expansion technique at time $t_0$ as that in the proof in \cite{Daubechies_Lu_Wu:2011}, we could approximate $|V_{x}^{(h)}(t_0,\xi)|^2$, $\Omega_{x}^{(h)}$ and $\Gamma_{x}^{(h)}$ by $|V_{x_0}^{(h)}(t_0,\xi)|^2$, $\Omega_{x_0}^{(h)}$ and $\Gamma_{x_0}^{(h)}$ with an error of order $\epsilon$, where $x_0$ is a linear chirp function
\[
x_0(t)=e^{i2\pi (\phi''(t_0) (t-t_0)^2/2+\phi'(t_0) (t-t_0)+\phi(t_0))}.
\]
To simplify the notation, denote
\[
\phi_0(t):=\phi''(t_0) (t-t_0)^2/2+\phi'(t_0) (t-t_0)+\gamma.
\]
Note that around $t_0$, the IF of $x_0$ is $\phi_0'(t)=\phi''(t_0) (t-t_0)+\phi'(t_0)$. By the relationship between the spectrogram and the WVD (\ref{Equation:KeyRelationshipWVDSTFT}), the spectrogram of $x_0$ at $t_0$ could be directly evaluated by
\begin{align}
&|V_{x_0}^{(h)}(t_0,\xi)|^2=\iint W_{x_0}(\tau,\eta)W_h(\tau-t_0,\eta-\xi)\ud \tau\ud \eta\nonumber\\
=&\,\int W_h(\tau-t_0, \phi_0'(\tau)-\xi)\ud \tau\,,
\end{align}
where $W_{x_0}(\tau,\eta)=\delta(\eta-\phi_0'(\tau))$ and $\delta$ is the Dirac measure. By \cite[p.222]{Flandrin:1999}, the reassignment rule in (\ref{RM:omega}) of $x_0$ could be equally expressed by
\begin{align}
\Omega^{(h)}(t_0,\xi)=&\,\frac{\iint \eta W_{x_0}(\tau,\eta)W_h(\tau-t_0,\eta-\xi)\ud \eta\ud \tau}{|V_{x_0}^{(h)}(t_0,\xi)|^2}\nonumber\\
=&\,\frac{\int \phi_0'(\tau) W_h(\tau-t_0, \phi_0'(\tau)-\xi)\ud \tau}{\int W_h(\tau-t_0, \phi_0'(\tau)-\xi)\ud \tau}\label{Proof:Theorem:Omega}\\
=&\, \phi'(t_0)+\phi''(t_0)\frac{\int \tau W_h(\tau,\phi_0'(\tau+t_0)-\xi)\ud \tau}{\int W_h(\tau, \phi'_0(\tau+t_0)-\xi)\ud \tau}\,,\nonumber
\end{align}
where the second equality comes from (\ref{Equation:KeyRelationshipWVDSTFT}) and the third equality comes from directly plugging $\phi_0'(\tau)$ inside.
Similarly, by \cite[p.222]{Flandrin:1999}, the group delay $\Gamma^{(h)}$ in (\ref{eq:gdf}) could be equally expressed by
\begin{align}
\Gamma^{(h)}(t_0,\xi)=&\,\frac{\iint \tau W_{x_0}(\tau,\eta)W_h(\tau-t_0,\eta-\xi)\ud \eta\ud \tau}{|V_{x_0}^{(h)}(t_0,\xi)|^2}\nonumber\\
=&\,\frac{\int \tau W_h(\tau,\phi_0'(\tau+t_0)-\xi)\ud \tau}{\int W_h(\tau, \phi_0'(\tau+t_0)-\xi)\ud \tau}\,,\label{Proof:Theorem:Gamma}
\end{align}
where the second equality holds by a direct calculation.
\end{proof}

We mention that this theorem is for the function having one generalized IMT function. For the case with more than one generalized IMT function, the proof is the same, except by taking into account the frequency separation condition (\ref{definition:adaptiveHarmonicMultiple}) and a well-chosen window support $T$ depending on $d$. Since the technique is the same as that in \cite{Daubechies_Lu_Wu:2011}, we skip the details. Based on Theorem \ref{Theorem:OmegaGamma}, we could discuss/define the  ``intrinsic latency'' for different TF analysis techniques.

\begin{figure}[]
    \centering
    \includegraphics[width=0.85\textwidth]{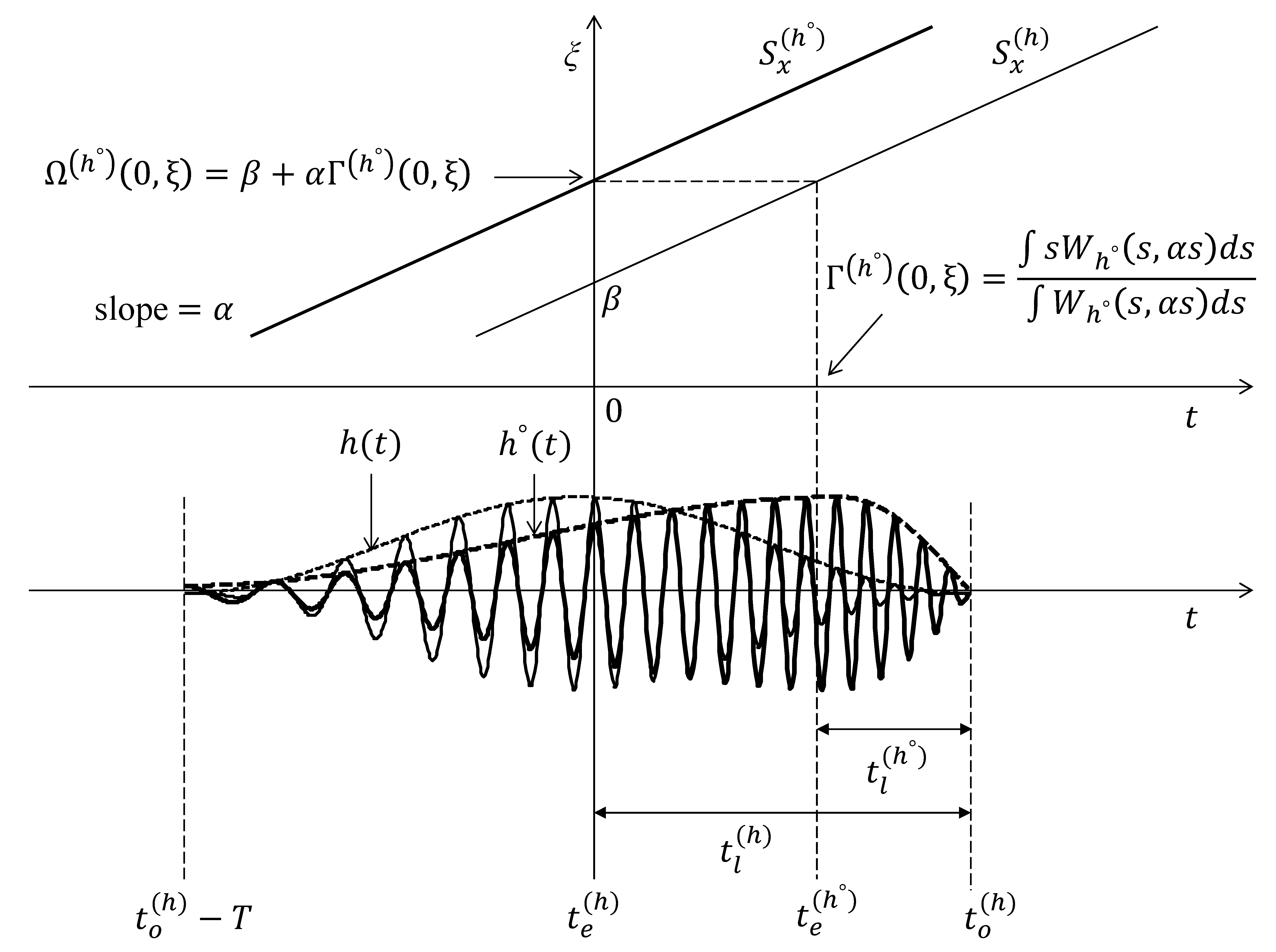}
    \caption{Conceptual illustration of short-time windowing and SST representations of the chirp signal with $\phi'(t)=\alpha t+\beta$ using a symmetric window (bold gray) and a asymmetric window (bold black). The symmetric and asymmetric windows are sketched in dashed lines. The definitions of the timestamps used in this paper, including the observation time $t_o$, the estimation time of the symmetric window $t^{(h)}_e$, and of the asymmetric window $t^{(h^{\circ})}_e$, are also illustrated respectively.}
    \label{fig:latency_demo}
\end{figure}

\subsubsection{STFT and SST}

To discuss and define the intrinsic latency for a TF analysis, we should make clear which information we expect to extract from the TFR. When the signal is oscillatory, the information of interest is how fast the signal oscillates, i.e., the IF, which could be extracted from the TFR. Based on Theorem \ref{Theorem:OmegaGamma}, we focus on quantifying how much lag the IF could be estimated.

First, we show that whether a given window is symmetric leads to different results. When $h$ is symmetric, due to the reflection symmetry, we have $W_h(\tau,c\tau)=W_h(-\tau,-c\tau)$ for all $\tau\geq 0$ and $c\geq0$. So
\[
\int \tau W_h(\tau,\phi''(t_0)\tau)\ud \tau=0.
\]
Thus, when $\int W_h(\tau,\phi''(t_0)\tau)\ud \tau\neq 0$,
\[
\Omega_x^{(h)}(t_0,\phi'(t_0))=\phi'(t_0)+O(\epsilon),
\]
which is the IF of $x$ at time $t_0$.
However, when $h$ is asymmetric, $\int \tau W_h(\tau,\phi''(t_0)\tau)\ud \tau$ may not be $0$, so
\[
\Omega^{(h)}(t_0,\phi'(t_0))=\phi'(t_0)+\phi''(t_0)\frac{\int \tau W_h(\tau,\phi''(t_0)\tau)\ud \tau}{\int W_h(\tau, \phi''(t_0)\tau)\ud \tau}+O(\epsilon),
\]
which may not be $\phi'(t_0)$ when $\int W_h(\tau, \phi''(t_0)\tau)\ud \tau\neq 0$. Thus, in general, at time $t_0$, by $\Omega^{(h)}(t_0,\phi'(t_0))$ we obtain the IF information at time $t_0+\frac{\int \tau W_h(\tau, \phi''(t_0)\tau)\ud \tau}{\int W_h(\tau, \phi''(t_0)\tau)\ud \tau}$ since
\begin{align}
\phi'_0\big(t_0+\frac{\int \tau W_h(\tau, \phi''(t_0)\tau)\ud \tau}{\int W_h(\tau,\phi''(t_0)\tau)\ud \tau}\big)=\,\phi'(t_0)+\phi''(t_0)\frac{\int \tau W_h(\tau,\phi''(t_0)\tau)\ud \tau}{\int W_h(\tau, \phi''(t_0)\tau)\ud \tau}.
\end{align}
Motivated by the above discussion, for a IMT function with IF $\phi'(t)$, we could consider the following quantity, called \textit{extrinsic latency of SST} and denoted as $t^{(h)}_{ex}$, associated with time $t$:
\begin{align}
t^{(h)}_{ex} := t+\frac{\int \tau W_h(\tau,\phi''(t_0)\tau)\ud \tau}{\int W_h(\tau, \phi''(t_0)\tau)\ud \tau}
\label{eq:center0}\,,
\end{align}
where the second term of (\ref{eq:center0}) indicates the temporal deviation of the obtained spectral information.
Note that this definition depends not only on the window but also on the signal, and hence the nomination.
%$\frac{\int \tau W_h(\tau,\phi''(t_0)\tau)\ud \tau}{\int W_h(\tau,\phi''(t_0)\tau)\ud \tau}$

The quantity $t^{(h)}_{ex}$ is supported by Theorem \ref{Theorem:OmegaGamma}, and is valid when the signal is composed of more than one IMT function. We could thus quantify the intrinsic latency of SST according to $t^{(h)}_{ex}$. It is well known that when the amplitude is slowly varying, the frequency reassignment rule well approximates the ridge of the spectrogram. While the IF information could be obtained by reading the spectrogram ridge, we could also define the \textit{extrinsic latency of STFT} by $t^{(h)}_{ex}$.

However, this quantity could not be evaluated since in general we do not have enough information about the signal's structure (e.g., AM, IF, single or multiple IMFs, etc.). We thus need a signal-independent quantity that reveals latency information.
To achieve this goal, we consider a slight modification of (\ref{eq:center0}) by taking ``all possible frequencies'' into account, and introduce the following definition for the \textit{intrinsic estimation time of STFT and SST}:
\begin{defn}
The \emph{intrinsic estimation time of STFT and SST} associated with time $t$ is defined as
\begin{align}
t^{(h)}_e &:=t+ \frac{\iint s W_h(s, \xi)\ud s\ud\xi}{\iint W_h(s, \xi)\ud s\ud \xi}
\,,\label{eq:center}
\end{align}
\end{defn}
Note that by a direct calculation, the estimation time satisfies
\[
t^{(h)}_e=t+\frac{\int \tau  |h(\tau)|^2 d\tau}{\int |h(\tau)|^2 d\tau}.
\]
%Clearly, $t^{(h)}_e=t+\int \tau  |h(t-\tau)|^2 d\tau$ if we assume $\|h\|_{L^2}=1$.
If $h$ is {\em symmetric}, then $t^{(h)}_e=t$; otherwise $t^{(h)}_e$ might deviate from $t$. Note that the last term of (\ref{eq:center}) is nothing but the gravity center of the window $h$. This result is similar to the time delay index (TDI) defined in \cite{florencio1991use}, but with a different normalization factor.

Equation (\ref{eq:center0}) can be further explained by two special case: a sinusoid-like signal and an impulse-like signal. If the signal is sinusoid-like, i.e., $\phi''(t)=0$, then $t^{(h)}_{ex}$ becomes
\begin{align}
t^{(h)}_{e0} &:=t+ \frac{\int s W_h(s, 0)\ud s}{\int W_h(s, 0)\ud s}
=t+\frac{\int \tau  h(\tau) d\tau}{\int h(\tau)d\tau}\,.\label{eq:center2}
\end{align}
If we consider $h(t)$ as a low-pass filter (LPF), then by (\ref{eq:gdf_filter}) we could find that $\frac{\int \tau  h(\tau) d\tau}{\int h(\tau)d\tau}=\Gamma^{(h)}(0)$, meaning that (\ref{eq:center2}) corresponds to the group delay of $h(t)$ at the center of its pass-band, i.e. zero frequency. On the other hand, when the signal is impulse-like, i.e., $|\phi''(t)|\gg 1$, then it is natural to consider the integration through ``all possible frequencies'' in the TFR. In this way, $t^{(h)}_{ex}$ converges to $t^{(h)}_e$.
In other words, for narrow-band signals, $t^{(h)}_{e0}$ is a good approximation for $t^{(h)}_{ex}$; while for the broad-band signals, $t^{(h)}_e$ is a good approximation for $t^{(h)}_{ex}$. %When the window is symmetric, both approximations become the same.
Table \ref{tab:comparison} summaries the differences of group delay and intrinsic latency between an FIR LPF and a window function.
For general signals that are supposed to be non-stationary and with unknown structure, we opt for using $t^{(h)}_{e}$ as the definition of intrinsic latency in performing TF analysis. Therefore we introduce the following definition of \textit{intrinsic latency of STFT and SST}.

\begin{defn}The \textit{intrinsic latency of STFT and SST} at time $t$, denoted as $t^{(h)}_l$, is defined to be the difference of the observation time and the estimation time:
\begin{align}
t^{(h)}_l:=t^{(h)}_o-t^{(h)}_e.
\label{eq: lat}
\end{align}
\end{defn}
It is clear that if $h$ is symmetric, then the intrinsic latency would be $T/2$. Also note that $t^{(h)}_o$ is strictly greater than $t^{(h)}_e$ so $t^{(h)}_l$ is strictly positive, unless we allow a discontinuous window function.

\begin{table}
\centering
\caption{Comparison of the group delay $\Gamma^{(h)}$ and intrinsic latency $t^{(h)}_l$ for low-pass filters and window functions in TF analysis}
\begin{tabular}{|c|c|}
\hline
 $h(t)$ as an FIR LPF & $h(t)$ as a window function\\
\hline
 \(\displaystyle \Gamma^{(h)}(\omega) =\left.\mathfrak{R} \frac{\widehat{\mathcal{T}h}(\omega)}{\hat{h}(\omega)}\right|_{\omega\approx 0}\)& \(\displaystyle \Gamma^{(h)}_{x}(t, \omega) = \mathfrak{R}\frac{V_x^{(\mathcal{T}h)}(t, \omega)}{V_x^{(h)}(t, \omega)}\)\\
\hline
 \(\displaystyle t^{(h)}_l=\frac{T}{2}-\frac{\int \tau  h(\tau) d\tau}{\int h(\tau) d\tau}\) & \(\displaystyle t^{(h)}_l=\frac{T}{2}-\frac{\int \tau  |h(\tau)|^2 d\tau}{\int |h(\tau)|^2 d\tau}\) \\
%\hline\hline
% Time-invariant & Time-variant \\
%%Stationarity of $x$ & Yes & No \\
%  & Depend on $x(t)$ \\
% Only in the pass-band & All possible frequencies \\
% Minimum-phase filter & Minimum-phase window \\
\hline
\end{tabular}
\label{tab:comparison}
\end{table}
%, $H(\omega)=\int h(t) e^{-j2\pi \omega t} dt$
%-\frac{\partial \Phi_{H}(\omega)}{\partial\omega}

\subsubsection{RM}
Last but not least, the intrinsic latency of the RM need to be discussed. By Theorem \ref{Theorem:OmegaGamma} and the term $g_\alpha\left(|t-s-\Gamma^{(h,\Theta)}_x(s,\xi)|\right)$ in (\ref{definition:RM}), the ``future IF information'' obtained by the frequency reassignment rule is ``forgotten'', since the temporal reassignment rule moves the future IF information back to its original temporal location. Therefore, we do not reduce the intrinsic latency and the {\em intrinsic latency of RM} is defined to be $T/2$. Note that this result re-emphasizes the difference between the minimal latency of TF analyses and the latency of a window. Depending on the nature of the chosen TF analysis, It is not always possible to reduce the latency.

\begin{figure*}[]
    \centering
    \includegraphics[width=\textwidth]{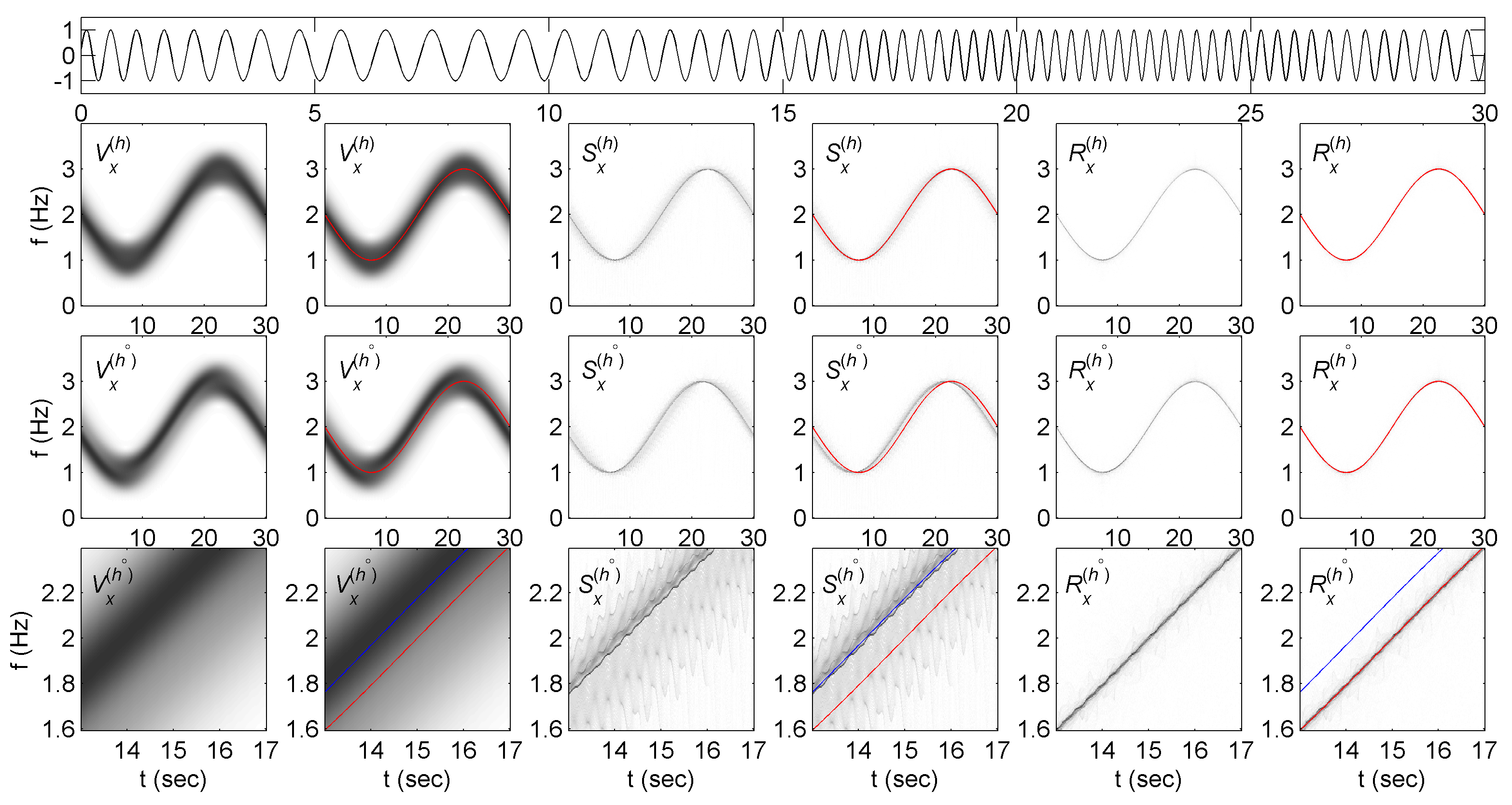}
    \caption{Reduction of the intrinsic latency using an asymmetric window for a synthetic input signal using a symmetric window $h(t)$ and an asymmetric window $h^{\circ}(t)$. Every TFR illustration is repeated twice, one with and one without the superposition of the ground truth. Top row: the signal. The second row, from left to right: STFT ($V^{(h)}_{x}$, 1st and 2nd columns), SST ($S^{(h)}_{x}$, 3rd and 4th columns) and RM ($R^{(h)}_{x}$, 5th and 6th columns). The third row, from left to right: $V^{(h^{\circ})}_{x}, S^{(h^{\circ})}_{x}$ and $R^{(h^{\circ})}_{x}$. Bottom row, from left to right: $V^{(h^{\circ})}_{x}, S^{(h^{\circ})}_{x}$ and $R^{(h^{\circ})}_{x}$ zoomed in on the region $t \in [13, 17]$ s and $f \in [1.6, 2.4]$ Hz. Red line: ground-truth IF. Blue line: ground-truth IF shifted ahead by the difference of intrinsic latency between $h$ and $h^{\circ}$.}
    \label{fig:stft_win}
\end{figure*}

\subsection{Numerical Examples}
Fig. \ref{fig:latency_demo} illustrates SST representations of the chirp signal with $\phi'(t)=\alpha t+\beta$ using a symmetric window $h(t)$ of length $T$ and an asymmetric window $h^{\circ}(t)$, as well as their positions of $t^{(h)}_o$ and $t^{(h)}_e$ under the definition derived from (\ref{eq:center}) and (\ref{eq: lat}). It is clear that the estimated time of $h^\circ$ is closer to the observation time compared with that of $h$, and therefore the intrinsic latency is smaller.

Fig. \ref{fig:stft_win} illustrates how an asymmetric window shifts the frequency trajectory of a TFR due to the reduced intrinsic latency. We consider a signal $x(t)=\cos(4\pi t+30\cos(\pi t/15))$ with the IF $2-\sin(\pi t/15)$. The signal is sampled at a sampling rate 100Hz during the period $[0, 30]$ seconds. The window size is 5 seconds. We illustrate three TFRs (i.e., STFT, SST, RM) using both $h$ and $h^\circ$, and every TFR illustration is shown twice, one with and one without the superposition of the ground truth IF. The ground-truth IF are superimposed to every figure in red line. The spectra of $h$ and $h^\circ$ are chosen to have the same magnitude spectrum.\footnote{The $h$ and $h^\circ$ adopted here are the flat-top and the minimum-phase (MP) flat-top windows, respectively. Detailed information will be introduced in Section \ref{sec:cosinewin}.} Notice that the time axis $t$ in this figure refers to the middle of the window function.

Fig. \ref{fig:stft_win} shows that, in comparison to the ground-truth IF, the ridges of $V^{(h^{\circ})}_{x}$ and $S^{(h^{\circ})}_{x}$ appear {\em earlier} than those of $V^{(h)}_{x}$ and $S^{(h)}_{x}$, respectively. This implies that the intrinsic latency using $h^{\circ}$ should be less than the one using $h$ by an amount that can be evaluated by Equations (\ref{eq:center}) and (\ref{eq: lat}); some calculations show that the intrinsic latency using $h^{\circ}$ is reduced by 0.84 second. In other words, the ridges of $V^{(h^{\circ})}_{x}$ and $S^{(h^{\circ})}_{x}$ should lead the ground-truth IF by this amount. This fact is illustrated in the bottom row by zoomed-in $V^{(h^{\circ})}_{x}$ and $S^{(h^{\circ})}_{x}$ to $t \in [13, 17]$. The superimposed blue lines are the ground-truth IF shifted ahead by exactly 0.84 second. We could see that the ridges are very close to the blue lines.

On the contrary, latency reduction is not observed in $R^{(h^{\circ})}_{x}$, mainly because the time-reassignment terms have included the contribution of latency caused by the asymmetry of the window, and the TFR at every frame are still reassigned back to the estimation time. In other words, RM exhibits $t^{(h)}_l=T/2$ for all kinds of windows; it is blind to see the change of latency by ``breaking the causality'' in time reassignment.

\section{Minimal latency window for TF analysis}
\label{sec:mlwin}

In this section, we construct minimum-latency windows suitable for the TF analysis. To simplify the discussion, we keep using $h$ to refer to the symmetric window and $h^\circ$ to refer to the asymmetric window with the spectral magnitude the same as that of $h$. The window is indexed by discrete time, namely, $h=[h(0)\,\,h(1)\,\,\ldots\,\,h(N-1)]\in\RR^N$, where $N\in\NN$ refers to the window size. Note that when the window $h$ is symmetric, $h(s-t)=h(t-s)$ in (\ref{eq: stft1}). Thus, to simplify the following discussion regarding minimum-phase (MP) signals, we consider $V^{(h)}_x(t, \omega) = \int x(s) h(t-s)e^{-i2\pi\omega s} ds$ and hence $h(0)$ corresponds to the observation time in the TF analysis\footnote{The ordering of the window index is reversed in the temporal axis since this definition can better facilitate the following discussion on the minimum-phase window.}. %

\subsection{Review of minimum-phase signals}

By definition, an MP signal $x\in\RR^N$ has all zeros and poles of its z-transform {\em inside} the unit circle, denoted as $\mathbb{T}=\{z\in\CC|\,|z|=1\}$, \cite{Oppenheim_Schafer:2009}.\footnote{Notice that some previous works adopt a loose definition stating that a MP signal is the one whose zeros are {\em on} or {\em within} the unit circle, for example \cite{damera2000design}. Although we adopt the strict definition in this paper, in most of the engineering problems the loose definition is enough. See the discussion on $\epsilon$-MP signal in the next paragraph.} As the phase of a MP signal is uniquely determined by its amplitude \cite{Oppenheim_Schafer:2009}, we can transform a signal (which might or might not be symmetric) without zeros {\em on} the unit circle in the z-domain to a MP signal while preserving its magnitude spectrum, by taking into account the property that the phase of the Fourier transform of a MP signal is the Hilbert transform of its log magnitude \cite{SASPWEB2011, Oppenheim_Schafer:2009, damera2000design}. Precisely, the MP signal could be constructed by $x_{\mathrm{MP}}=\mathcal{M}x:=\mathcal{F}^{-1}(|\mathcal{F}x| e^{i \mathcal{H}\left[\log |\mathcal{F}x|\right]})$, where $\mathcal{F}$ is the Fourier transform, $\mathcal{H}$ is the Hilbert transform and $\mathcal{F}^{-1}$ is the inverse Fourier transform. We name $\mathcal{M}$ as the {\em minimum-phase transform}.
The MP transform has been taken as a method in designing MP filters with low-latency properties \cite{damera2000design}: we can convert a linear-phase (symmetric) filter to an MP filter and let both have the same amplitude transfer function by means of the MP transform.
An important property of a MP signal relevant to us is that among all signals with the same spectral magnitude, the MP signal has the maximal concentration of energy toward the observation time \cite[p.290]{Oppenheim_Schafer:2009}:
\begin{equation}
\sum^{k}_{n=0}|x_{\MP}(n)|^2 \geq \sum^{k}_{n=0}|x(n)|^2
\label{eq:mp_energy}
\end{equation}
for all $k=0,1,\cdots,N-1$. This property exactly facilitates our requirement of deriving a low-latency and asymmetric window with known spectrum information for the TF analysis.

One issue of applying the MP transform is that the z-transforms of several widely applied symmetric windows have zeros on $\mathbb{T}$. To resolve this issue, an intuitive approach taken in practice is the {\em $\epsilon$-perturbation}; that is, we could construct an MP window $h_{\eMP}$ from $h$ by
\begin{align}\label{Definition:epsilonMinimumPhase}
h_{\eMP}=\mathcal{M}_\epsilon h:=\mathcal{F}^{-1}((|\mathcal{F}h|+\epsilon) e^{i \mathcal{H}\left[\log (|\mathcal{F}h|+\epsilon)\right]}),
\end{align}
\noindent where $0<\epsilon\ll 1$ is chosen so that $|\mathcal{F}h|+\epsilon\neq 0$ on $\mathbb{T}$. We call $\mathcal{M}_\epsilon$ the {\em $\epsilon$-MP transform}, and call $h_{\eMP}$ the {\em $\epsilon$-minimum-phase} associated with $h$. It is clear that for a fixed $\epsilon$, $h_{\eMP}$ is unique, by the uniqueness property of a MP window. While the $\epsilon$-MP transform was applied in the past in designing asymmetric windows, however, it is not clear if it satisfies the property that the energy is maximally concentrated toward the observation time. Below we provide a positive answer.

\begin{thm}
For $\epsilon> 0$, $h_{\eMP}$ is maximally concentrated toward the index $0$ up to $\epsilon$; that is, $\sum^{k}_{n=1}|h_{\eMP}(n)|^2 \geq \sum^{k}_{n=1}|h(n)|^2-\epsilon$ for all $k$.
\end{thm}
\begin{proof}
Denote $\mathcal{F}h:=|\mathcal{F}h| e^{i \text{arg}\mathcal{F}h}$, where $\text{arg}$ means taking the continuous phase of $\mathcal{F}h$, and denote $\tilde{h}:=\mathcal{F}^{-1}((|\mathcal{F}h|+\epsilon) e^{i \text{arg}\mathcal{F}h})$. It is clear that $\|\tilde{h}-h\|_{L^2}\leq \sqrt{2\pi} \epsilon$ and $\|\tilde{h}-h\|_{L^\infty}\leq 2\pi \epsilon$. By definition, $h_{\eMP}=\mathcal{M} \tilde{h}$. Since $|\mathcal{F}\tilde{h}|\neq 0$, $h_{\eMP}$ is the unique MP companion of the window $\tilde{h}$. By (\ref{eq:mp_energy}) and the above relationship between $h$ and $\tilde{h}$, we conclude the claim.
\end{proof}
In the following, to simplify the discussion, we would drop the subscript $\epsilon$ and use $\mathcal{M}$ to denote the $\epsilon$-MP transform unless there is a danger of confusion. Next, we show that $h_{\MP}$ is a good candidate for the minimal latency TF analysis as its estimation time is closest to the observation time among all windows with the same magnitude spectrum.

\begin{thm}
Take two window functions $h_{\MP}(n)$, $h(n)\in\RR^N$, with $|\mathcal{F}h_{\MP}|=|\mathcal{F}h|$, and some zeros of the z-transform of $h$ are outside $\mathbb{T}$. Then, by definition in (\ref{eq:center}) and (\ref{eq: lat}), $t^{(h_{\MP})}_l< t^{(h)}_l$.
\label{thm3-1}
\end{thm}
\begin{proof}
The discrete-time version of (\ref{eq:center}) is represented as
\begin{equation}
t^{(h)}_e=\frac{\sum^{N}_{n=1}(N+1-n)|h(n)|^2}{\sum^{N}_{n=1}|h(n)|^2}\,.
\end{equation}
Proving $t^{(h_{\MP})}_l\leq t^{(h)}_l$ is equivalent to proving $t^{(h_{\eMP})}_e\geq t^{(h)}_e$, as the two windows are compared in accordance with the same observation time. Since $h_{\MP}$ and $h$ have the same magnitude spectrum, we have $\sum^{N}_{n=1}|h_{\MP}(n)|^2=\sum^{N}_{n=1}|h(n)|^2$. Therefore, it suffices to prove that
\begin{equation}
\sum^{N}_{n=1}(N+1-n)|h_{\MP}(n)|^2\geq \sum^{N}_{n=1}(N+1-n)|h(n)|^2\,. \nonumber
\end{equation}
To prove this, denote $z_l$, $l=1,\ldots,m$, the zeros of the z-transform of $h$, denoted as $H$, outside $\mathbb{T}$; that is $|z_l|>1$. Then for all $k=1,2,\cdots,N$ and a fixed $z_l$, the following energy concentration property is satisfied:
\begin{equation}
\sum^{k}_{n=1}|h_{\MP}(n)|^2-\sum^{k}_{n=1}|h(n)|^2\geq (1-1/|z_l|^2)||s_l(k)|^2, \nonumber
\end{equation}
where $s_l$ is another MP sequence whose z-transform, denoted as $S_l(z)$, satisfies $H(z)=S_l(z)(z_l-z)$.
Denote $q(n):=|h_{\MP}(n)|^2-|h(n)|^2$. Note that
\[
\sum^{k}_{n=1}q(n)\geq(1-1/|z_l|^2)||s_l(k)|^2> 0
\]
for all $k=1,2,\cdots,N$. Then, by a direct calculation we have
\begin{align}
&\sum^{N}_{n=1}(N+1-n)|h_{\eMP}(n)|^2 - \sum^{N}_{n=1}(N+1-n)|h(n)|^2 \nonumber\\
=&\, \sum^{N}_{n=1}(N+1) q(n) - \sum^{N}_{n=1}n q(n) \nonumber\\
=&\, (N+1)\sum^{N}_{n=1}q(n)-\left(\sum^{N}_{n=1}q(n)+\sum^{N}_{n=2}q(n)+\cdots+\sum^{N}_{n=N}q(n)\right) \nonumber\\
=&\, \left(\sum^{N}_{n=1}q(n)-\sum^{N}_{n=1}q(n)\right)+\left(\sum^{N}_{n=1}q(n)-\sum^{N}_{n=2}q(n)\right)+\cdots+\left(\sum^{N}_{n=1}q(n)-\sum^{N}_{n=N}q(n)\right) +\sum^{N}_{n=1}q(n) \nonumber\\
=&\, 0+\sum^{1}_{n=1}q(n)+\sum^{2}_{n=1}q(n)+\cdots+\sum^{N-1}_{n=1}q(n)+\sum^{N}_{n=1}q(n) \nonumber\\
\geq&\,(1-1/|z_l|^2)|\sum^{N}_{k=1}|s_l(k)|^2>0 \nonumber\,,
\end{align}
which shows the claim.
\end{proof}

Note that the conclusion $t^{(h_{\MP})}_l< t^{(h)}_l$ indicates that the MP window will reduce the intrinsic latency for TF analysis when there is at least one root outsize $\mathbb{T}$. In sum, we have $t^{(h)}_l=(N-1)\tau/2$ and $t^{(h^{\circ})}_l< (N-1)\tau/2$, where $\tau$ is the sampling period in the digitization of a given signal.

\begin{figure}[]
    \centering
    \includegraphics[width=\columnwidth]{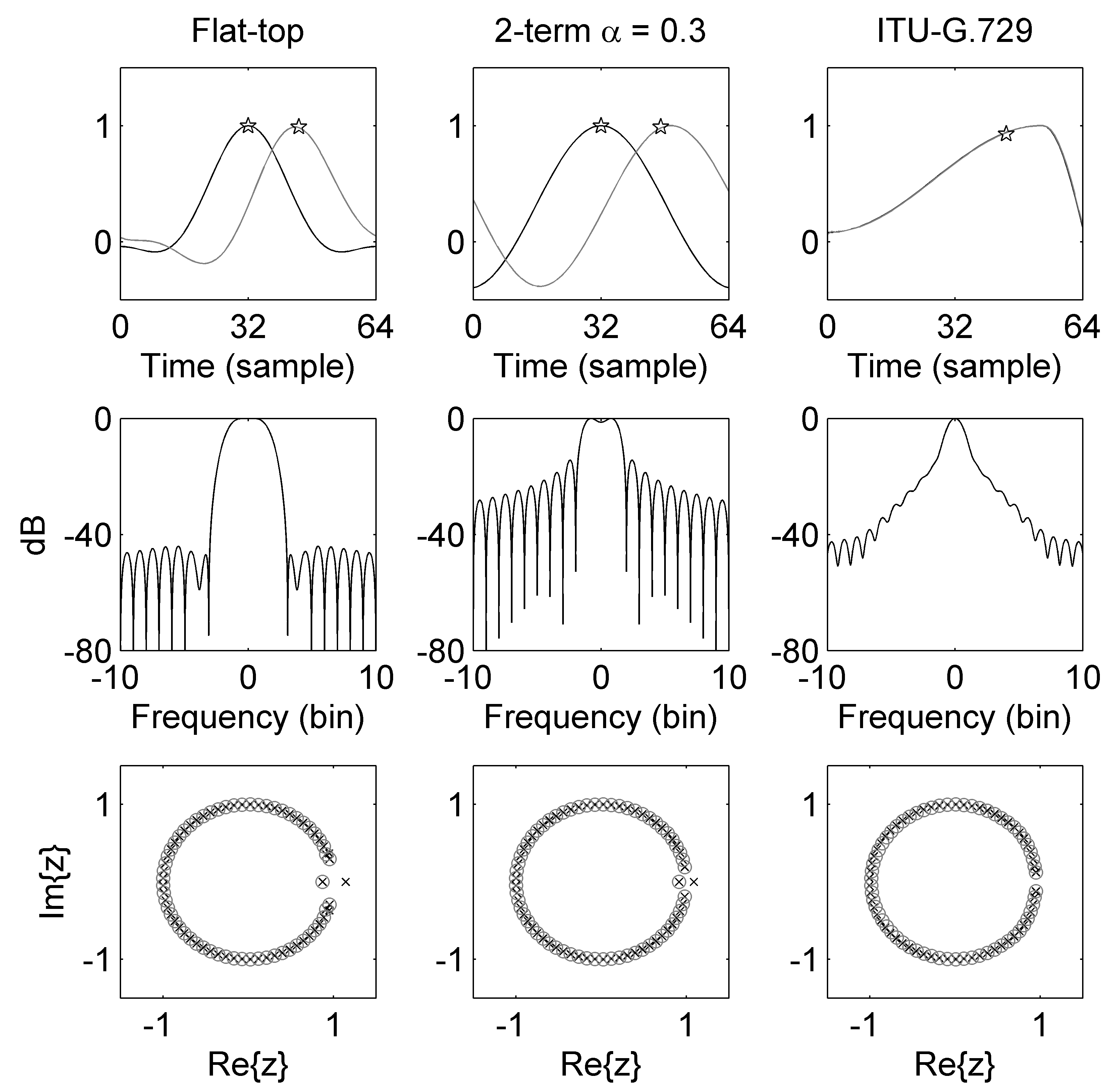}
    \caption{Illustration of the asymmetric windows (top), their dB-scaled magnitude spectra (middle) and zeros (bottom). From left to right: flat-top window (the black line is before and the gray line is after the MP transform), the 2-term cosine series window with $\alpha_0=0.3$ (the black line is before and the gray line is after the MP transform), and the ITU-T G.729 window. These windows are in $\RR^{65}$. Asterisk: the position of the estimation time. In the bottom figure, zeros of the window before the MP transform are marked as crosses and zeros of the window after the MP transform are marked as circles. Notice that these windows do not have poles.}
    \label{fig:mp_demo}
\end{figure}

\subsection{Design of minimal latency window}

While there are various ways in designing a symmetric window for different purposes \cite{Harris:1978}, motivated by Theorem \ref{thm3-1}, we could apply the $\epsilon$-MP transform to a well designed symmetric window $h$ and expect to obtain an asymmetric window $h^\circ$ fulfilling our needs.
However, there are several questions regarding this approach, where the most important one might be: given a chosen symmetric window $h$ with desired properties, how to confirm whether all the zeros of its z-transform are located on $\mathbb{T}$? Note that when all the zeros are located on $\mathbb{T}$, essentially the $\epsilon$-MP transform does not change the symmetry property of $h$, and hence the intrinsic latency cannot be reduced.

This question is equivalent to asking whether all the zeros of a {\em palindromic polynomial} are located on $\mathbb{T}$. To answer this question, one approach is simply applying root finding algorithm to find all roots with enough good accuracy, and checking whether the norm of all roots are close enough to one. The other approach is, more precisely, establishing an exact testing procedure by introducing rigorous analysis on the sufficient and necessary condition for a palindromic polynomial to have all its zeros on $\mathbb{T}$ \cite{Suzuki:2016}. Our pilot study has shown that both approaches work efficiently in practice. Since the first method is rather straightforward, in the followings we will only introduce the second approach in detail.

%; for example, how to guarantee if $h^\circ$ could help reduce the intrinsic latency, how to design the MP window with the minimal intrinsic latency, etc. In this subsection, we study the first question

Recall that a nonzero polynomial $P(x) = c_0x^n + c_1x^{n-1} + \ldots + c_{n-1}x + c_n$ with real coefficients is a \textit{palindromic polynomial} (respectively \textit{anti-palindromic} polynomial) of degree $n$ if $c_0\neq  0$ and $c_k = c_{n-k}$ (respectively $c_k=-c_{n-k}$) for every $0\leq k \leq n$. Note that the z-transform of a real and symmetric (respectively anti-symmetric) window $h=[h(0),\ldots,h(N)]\in \RR^{N+1}$, where $h(0)\neq 0$, is a palindromic (anti-palindromic) polynomial of degree $N$. In this study, we focus on the symmetric windows.
There are several good properties of palindromic and anti-palindromic polynomials.
\begin{enumerate}
\item[(P1)] All roots of a palindromic or an anti-palindromic polynomial appear pairwisely; that is, if $\lambda\in\CC$ is a root, then $1/\lambda$ is also a root. The converse is also true.
\item[(P2)] If a palindromic (respectively anti-palindromic) $P$ is of odd degree, we could find a palindromic polynomial $\tilde{P}$ of even degree so that $P(x)=(x+1)\tilde{P}(x)$ (respectively $P(x)=(x-1)\tilde{P}(x)$); if an anti-palindromic $P$ is of even degree, we could find a palindromic polynomial $\tilde{P}$ of even degree so that $P(x)=(x^2-1)\tilde{P}(x)$.
\end{enumerate}

By (P1), we know that if a window's z-transform has a non-unitary root, then the $\epsilon$-MP transform could help to convert a chosen symmetric window to an asymmetric but not anti-symmetric window.  By (P2), the question is reduced to ask when a palindromic polynomial of even degree has all its zeros located on $\mathbb{T}$.
While there are several if and only if conditions for this question, to the best of our knowledge, they either depend on the knowledge of finding the roots of the derivative of the palindromic polynomial \cite[Section 7.5]{Suzuki:2016} or decomposing the polynomial into a product of a sequence of quadratic polynomials. These conditions are however not directly calculable for the practical purpose. We thus consider the following condition based solely only on reading the coefficients of a palindromic polynomial.

We need to introduce the following quantities. The first set of notations are for matrices.
For $n,m\in\NN$, denote $0_{n\times m}$ to be the zero matrix of size $n\times m$, $e_j\in\RR^{n}$ to be the unit vector with the $j$-th entry $1$, $I_n$ to be the $n\times n$ identity matrix, and $\tilde{I}_n$ to be the anti-diagonal matrix of size $n\times n$ with all anti-diagonal entries $1$.
For $k\in\NN$, define
\begin{align}
&J^{(1)}_k:=[I_k| 0_{k\times 1}]\in \RR^{k\times(k+1)},\nonumber\\
&J^{(2)}_k:=[0_{k\times 1}| I_k]\in \RR^{k\times(k+1)},\nonumber\\
&J^{(3)}_k:=\left[\begin{array}{cc}0_{1\times k} \\ \hline \tilde{I}_k  \end{array}\right]\in\RR^{(k+1)\times k}\nonumber.
\end{align}
Then, for an indeterminate $x\in\RR$, define the following matrices
\begin{align}
&P_0:=I_2,\, Q_0:=\begin{bmatrix} 1 & 1 & 0 & 0\\ 0 & 0 & 1 & -1\end{bmatrix}\nonumber\\
&P_1(x):=\begin{bmatrix}1 & 0 & 0 & 0\\ 0 & 1 & 0 & 0 \\ 0 & 0 & 1 & 0\\ 0 & 1 & 0 & -x \end{bmatrix},\, Q_1:=\begin{bmatrix} 1 & 0 & 1 & 0 & 0 & 0\\0 & 1 & 0 & 0 & 0 & 0\\0 & 0 & 0 & 1 & 0 & -1 \\ 0 & 0 & 0 & 0 & 0 & 0
\end{bmatrix};\nonumber
\end{align}
For $k\geq 2$, $P_k(x)$ and $Q_k$ are defined blockwisely as
\begin{align}
&P_k(x):=\begin{bmatrix} V_k^+ & 0\\ 0 & V_k^- \\ J^{(1)}_k & -xJ^{(2)}_k\end{bmatrix},\,\,Q_k:=\begin{bmatrix} W_k^+ & 0\\ 0 & W_k^- \\ 0_{k\times (k+2)} & 0_{k\times (k+2)}\end{bmatrix},\nonumber
\end{align}
where
\begin{align}
V_k^+:=\left\{\begin{array}{ll}
\big[ {J^{(1)}}^T_{\frac{k+1}{2}} \big| J_{\frac{k+1}{2}}^{(3)}
\big]&\mbox{ when }k\mbox{ is odd}\\
\big[ I_{\frac{k}{2}+1}  \big| J_{\frac{k}{2}}^{(3)}
\big]&\mbox{ when }k\mbox{ is even},
\end{array}
\right. \nonumber
\end{align}
\begin{align}
V_k^-:=\left\{\begin{array}{ll}
\big[ I_{\frac{k+1}{2}} \big| 0_{\frac{k+1}{2}\times 1} \big| -J_{\frac{k-1}{2}}^{(3)}
\big]&\mbox{when }k\mbox{ is odd}\\
\big[ I_{\frac{k}{2}+1} \big|  -J_{\frac{k}{2}}^{(3)}
\big]&\mbox{when }k\mbox{ is even},
\end{array}\right.\nonumber
\end{align}
\begin{equation}
W_k^+=[V_k^+|e_1]\mbox{ and }W_k^-=[V_k^-|-e_1];
\end{equation}
that is, $W_k^{\pm}$ comes from adding $\pm e_1$ on the right hand side of $V_k^{\pm}$.
The second set of notations are used to check the locations of the roots. Take $K\in\NN$. For $z\in \RR^{4K+2}$ indeterminate, define
\begin{equation}
v^{(K)}_0(z)=z\in\RR^{4K+2}.
\end{equation}
We then recursively define for $n=1,2,\ldots$
\begin{equation}
m_{2K-n}(z):= \frac{e_1^Tv^{(K)}_{n-1}(z)+e_{2K-n+2}^Tv^{(K)}_{n-1}(z)}{e_{2K-n+3}^Tv^{(K)}_{n-1}(z)+e_{4K-2n+4}^Tv^{(K)}_{n-1}(z)}
\end{equation}
and hence
\begin{equation}
v^{(K)}_n(z):=P_{2K-n}(m_{2K-n}(z))^{-1}Q_{2K-n}v^{(K)}_{n-1}(z).
\end{equation}

With the above notations, we could state the following theorem:
\begin{thm}\label{TheoremSuzuki:2016}
\cite[Theorem 1.7]{Suzuki:2016}
Let $K \geq 1$ and $q > 1$. Let $P(x)=\sum_{l=0}^{2K}h(l)x^l$ be a palindromic polynomial of degree $2K$ with real coefficients. Take $z=[a^T\,\, b^T]^T\in\RR^{4K+2}$, where $$
a:=\begin{bmatrix}h(0)\\h(1)\\\vdots\\h(K-1)\\h(K)\\h(K+1)\\\vdots\\h(2K-1)\\h(2K)\end{bmatrix},\,\,b:=\begin{bmatrix}h(0)\log q^K\\h(1)\log q^{K-1}\\\vdots\\h(K-1)\log q\\0\\-h(K+1)\log q\\\vdots\\-h(2K-1)\log q^{K-1}\\-h(2K)\log q^K\end{bmatrix}\in\RR^{2K+1}
$$
Then, all roots of $P(x)$ are located on $\mathbb{T}$ if and only if
$0<m_{2K-n}(z)<\infty$ for every $1\leq n \leq 2K$;
\end{thm}
The proof is based on studying an inverse problem for a class of canonical systems and can be found in \cite{Suzuki:2016}.
Based on the above discussion, we propose the following systematic way of designing asymmetric windows with minimal latency instead of relying on heuristics.
\begin{enumerate}
\item Select a well-designed real and symmetric window $h$ with the desired spectral properties;
\item Confirm if all roots are on $\mathbb{T}$ by utilizing Theorem \ref{TheoremSuzuki:2016};
\item If roots are not all on $\mathbb{T}$, apply the $\epsilon$-MP transform to convert $h$ to a new window $h^\circ$.
\end{enumerate}
In this case, based on Theorem \ref{thm3-1}, we know that $h^\circ$ has the intrinsic latency for TF analysis smaller than that of $h$, and it has the same spectrum information as that of $h$, including the main-lobe width and first side lobe level, etc., all of which are known from the original symmetric window.

\subsection{Example: short cosine series windows with MP transform}
\label{sec:cosinewin}

\subsubsection{Review of short cosine series windows}
To illustrate how the window design scheme works, we consider the \emph{short cosine series windows} as an example, which is arguably the most thoroughly used type of windows in TF analysis. The general form of a \emph{K}-term cosine series window with length of $N$ is
\begin{equation}\label{Definition:CosineWindow}
h^{[\alpha_0,\alpha_1,\ldots,\alpha_{K-1}]}\left(n\right)=\sum^{K-1}_{k=0} (-1)^{k}\alpha_k c_k(n),
\end{equation}
where $n=0,1,\cdots,N-1$, $c_k(n):=\cos(2\pi kn/N)$, and $\alpha_k>0$ so that $\sum^{K-1}_{k=0} \alpha_k = 1$ is satisfied. Obviously, all of these windows are symmetric.
Well known examples include the \emph{2-term cosine series windows} like Hamming window and Hann window; that is, $K=2$ in (\ref{Definition:CosineWindow}),
$\alpha_0=0.54$ for the Hamming window and $\alpha_0=0.50$ for the Hann window, and \emph{3-term cosine series windows}, that is, $K=3$, $\left[\alpha_0, \alpha_1, \alpha_2\right]=\left[0.42, 0.50, 0.08\right]$ for the Blackman window, and $\left[\alpha_0,\alpha_1,\alpha_2\right]=\left[0.28, 0.52, 0.20\right]$ for the flat-top window. The characteristics of these windows can be found in \cite{Harris:1978,Nuttall:1981,Rozman_Kodek:2007}.

Note that by Theorem \ref{TheoremSuzuki:2016}, the zeros of the z-transform of the Hamming, Hann and Blackman windows are all on the unit circle, so none of these windows can be transformed to an asymmetric by the MP transform.
On the other hand, by Theorem \ref{TheoremSuzuki:2016}, the flat-top window and the 2-term cosine windows with $\alpha_0<0.50$ have at least one zero outside the unit circle, and therefore can be transformed into asymmetric windows with the smallest intrinsic latency guarantee by Theorem \ref{thm3-1}. This can be clearly seen in the first two columns of Fig. \ref{fig:mp_demo}, where the asymmetric window (gray line) has the same magnitude spectrum as the symmetric window, but all zeros of the symmetric window are transformed into the unit circle.
We refer to this window as the {\em MP flat-top window} in the rest of this paper.

It is also worth mentioning the ITU-T G.729 standard, a VoIP (Voice over IP) standard now widely used in Internet telephony like Skype, Google Talk, etc. An important feature is that it offers a low latency data transfer. The standard envisaged to reduce the intrinsic latency from 15 ms to around 5 ms when the window size is 30 ms \cite{ITUG729}. The window used in ITU-T G.729 standard, called the hybrid hamming-cosine window, is represented as
\begin{equation}
h_{I}(n) =
\begin{cases}
\cos\left(\frac{2\pi n}{2N/3-1}\right), & 0 \leq n \leq \frac{N}{6}-1, \\
0.54-0.46\cos\left(\frac{2\pi(n-N/6)}{5N/3-1}\right), & \frac{N}{6} < n \leq N-1. \\
\end{cases}
\end{equation}
The rightmost column of Fig. \ref{fig:mp_demo} illustrates the window. The ITU-T G.729 window has an advantage for its main-lobe being narrower than others; its 3-dB bandwidth is only 1.2 bins, while the 3-dB bandwidth of the flat-top window is 2.9 bins. Its highest side-lobe level is at -29.62 dB. Intriguingly, the zeros of the ITU-T G.729 window are very close to (although not on) the unit circle; a MP transform converts the ITU-T G.729 window to another one which is quite similar to itself. Although not mentioned explicitly in the literature, this fact indicates that the ITU-T G.729 window not only has adequate temporal continuity, bandwidth and side-lobe level, but also behaves like an MP window.

\subsubsection{Optimized 2-term cosine window}
Given the model of a $K$-term cosine window, the intrinsic latency could be optimized by choosing the optimal coefficient. Consider $K=2$. By Theorem \ref{TheoremSuzuki:2016} we found that when $\alpha_0\geq 0.5$, all zeros of the z-transform of the window are on the unit circle, so the intrinsic latency lower than $N/2$ cannot be constructed through the MP transform. On the other hand, since an asymmetric window with the same spectral magnitude can be constructed only when $\alpha_0< 0.5$, to achieve the minimal latency, we consider the optimization problem:
\begin{equation}
\alpha^{(opt)}_0=\argmin_{\alpha_0\in(0,0.5)}t^{(h^{[\alpha_0\,\,\, (1-\alpha_0)]}_{{\eMP}})}_l\,.
\end{equation}
By sweeping over $0<\alpha_0\leq 0.5$, we found that when $\alpha_0=0.30$ we have a minimal intrinsic latency $t^{(h^{[0.3\,\, 0.7]}_{{\eMP}})}_l=0.285 N$. This MP-transformed window has the first side-lobe level of only -14.37 dB, but all the zeros are within the unit circle, as illustrated in the third column of Fig. \ref{fig:mp_demo}.

Table \ref{tab: comp} summarizes the observation time ($t^{(h)}_o$), estimation time ($t^{(h)}_e$) and intrinsic latency ($t^{(h)}_l$) of several asymmetric windows associated with $t=0$, all of which are computed by Equations (\ref{eq:center}) and (\ref{eq: lat}). The intrinsic latencies of these asymmetric windows are reduced from $N/2$ (symmetric case) to less than $N/3$.

\begin{table}
\centering
\caption{Observation time, estimation time and intrinsic latency of selected windows in $\RR^N$ and centered at $n=0$.}
\begin{tabular}{|l|c|c|c|}
\hline
$h$ & $t^{(h)}_o$ & $t^{(h)}_e$ & $t^{(h)}_l$ \\
\hline
\hline
All symmetric windows & $0.500N\tau$ & 0 & $0.500N\tau$\\
MP Flat-top & $0.500N\tau$ & $0.169N\tau$ & $0.331N\tau$ \\
MP 2-term, $\alpha=0.30$ & $0.500N\tau$ & $0.215N\tau$ & $0.285N\tau$ \\
ITU-T G.729 & $0.500N\tau$ & $0.179N\tau$ & $0.321N\tau$ \\
\hline
\end{tabular}
\label{tab: comp}
\end{table}

\section{Application on Real-time Musical Onset Detection}
\label{sec:application}

\subsection{Introduction}

Musical onset detection is the task of finding the beginning timestamp of every meaningful music event (i.e., notes, drum hits, etc.) in a given music excerpt \cite{bello2005tutorial,dixon2006onset,holzapfel2010three}. Real-time onset detection is a fundamental building block in interactive music systems, with application to musicology research, music performance, education, entertainment and others \cite{mcfee2014better,robertson2007b, benetos2013automatic,von2014cmmsd,lartillot2013more}.
Real-time onset detection requires light computation and low latency, and the latter means that the delay between input and output of the system should be minimized \cite{bock2012evaluating,bock2012online,glover2011real}. The required latency varies in application. An accepted latency is generally less than 10 ms, which is challenging because of the TFR distortion caused by the strong effect of spectral leakage induced by the short window function  \cite{brandt1998low,annett2014low,chaudhary2001perceptual,pardue2014lowlatency}.

The {\em spectral flux} (SF) is arguably the most widely-used onset detection method. It nicely measures the total incremental amount of the TFR magnitude (e.g., the spectrogram) along every consecutive time \cite{bock2013maximum, su2014power}. Onset events are determined by the peaks of the curve representing the SF at each time; see (\ref{eq:PSSF}) for the definition.
Besides SF, recently more advanced features are introduced, like multi-resolution STFT \cite{bock2012online}, wavelet transform \cite{marchi2014multi}, correntropy \cite{chang2014pairwise}, harmonic cepstrum \cite{heo2013note}, constrained linear reconstruction \cite{liang2015musical}, and supervised method like neural networks \cite{eyben2010universal,marchi2014multi,bock2012online}, to name but a few.
In this study, we opt for the well-studied SF method in the experiment, as this is our first attempt at investigating the effect of asymmetric windows for real-time onset detection.

\subsection{Algorithms}

An onset detection system contains mainly three parts: feature extraction,  onset detection function (ODF) computation and onset decision. We consider two feature representations, the STFT ($V^{(h)}_x$) and the SST ($S^{(h)}_x$). The ODF computation and onset decision stages are described as follows.

\subsubsection{Onset detection function}

We use the power-scaled SF as the ODF algorithm \cite{su2014power}. Fix the power $p\geq0$ and introduce a parameter $\mu\geq1$ to control the number of frames apart where the difference is computed. The SF is combined with a maximum filtering process \cite{bock2013maximum} with the parameter $\eta\geq0$:
\begin{equation}
\label{eq:PSSF}
\mathrm{SF}(n)=\sum^{N-1}_{k=0}H\left(|X(n,k)|^p-\max_{k-\eta \leq k' \leq k+\eta}|X(n-\mu,k')|^p \right )\,.
\end{equation}
where $X$ is the chosen TFR, $N$ is the size of frequency axis grid of STFT, $n$ and $k$ are the time and frequency indices respectively, and $H(x):=(x+|x|)/2$, $x\in\RR$, is the half-wave rectifier function. The {\em maximum filtering} process is useful in suppressing unwanted peaks, especially when the frequency contours in the signal have some slight variation; the most evident example is the \emph{vibrato} \cite{bock2013maximum}. In this work, we set $p = 0.5$, $\eta = 1$ and $\mu = 3$ to compute the ODF.

\subsubsection{Peak picking}

After obtaining the refined ODF, we refer to a real-time peak picking process to find the onset \cite{bock2013maximum}. The onset is determined by the following conditions:
\begin{align}
\mathrm{SF}\left(n\right) =&\, \max_{m=n-30\ \mathrm{ms}\ldots,n}\left(\mathrm{SF}\left(m\right)\right)\nonumber \\
\mathrm{SF}\left(n\right) \geq&\, \mathrm{mean} \left(\mathrm{SF}\left(n-150\ \mathrm{ms}\right),\ldots,\mathrm{SF}\left(n\right)\right)+\delta\,,
\end{align}
where the threshold $\delta$ determines the minimal salience in which a peak is regarded as an onset event. In our experiment we set $\delta=0.15$.

\subsection{Evaluation}

We evaluate our methods on a subset of the MIDI Aligned Piano Sounds (MAPS) database \cite{Emiya08thesis, lee12tmm}. The subset contains 30 piano pieces recorded by an upright Yamaha Disklavier piano. The annotation data of MAPS include the onset of every note. There are more than 10,000 onsets in this dataset.

Following previous work, we evaluate the performance by the F-score. The F-score is defined as $F := 2PR/(P+R)$, with Precision, \emph{P}, and Recall, \emph{R}, being computed from the number of correctly detected onsets $N_{tp}$, the number of false alarms $N_{fp}$, and the number of missed onsets $N_{fn}$, where $P:=N_{tp}/(N_{tp}+N_{fp})$ and $R:=N_{tp}/(N_{tp}+N_{fn})$. A detected onset is correct when it is located within a tolerance of $\sigma=\pm50$ ms around the ground truth annotation.\footnote{Note that the tolerance $\sigma$ is greater than the intrinsic latency in most cases. However, this is necessary because there are inevitable deviations (e.g., perceptual difference among the annotators, deviation between mechanical strikes and sound generations, etc.) of the ground truth, and such deviations prevent an onset detection algorithm from being arbitrarily accurate (more detailed discussion on the deviation of music annotation can be seen in \cite{su2015escaping}). Thus, given that the onset labels are distributions around the true onsets, using a large tolerance window is considered a meaningful evaluation method.} Detailed information of evaluation could be found in \cite{holzapfel2010three}.

\begin{figure}
\centering
\includegraphics[width=\columnwidth]{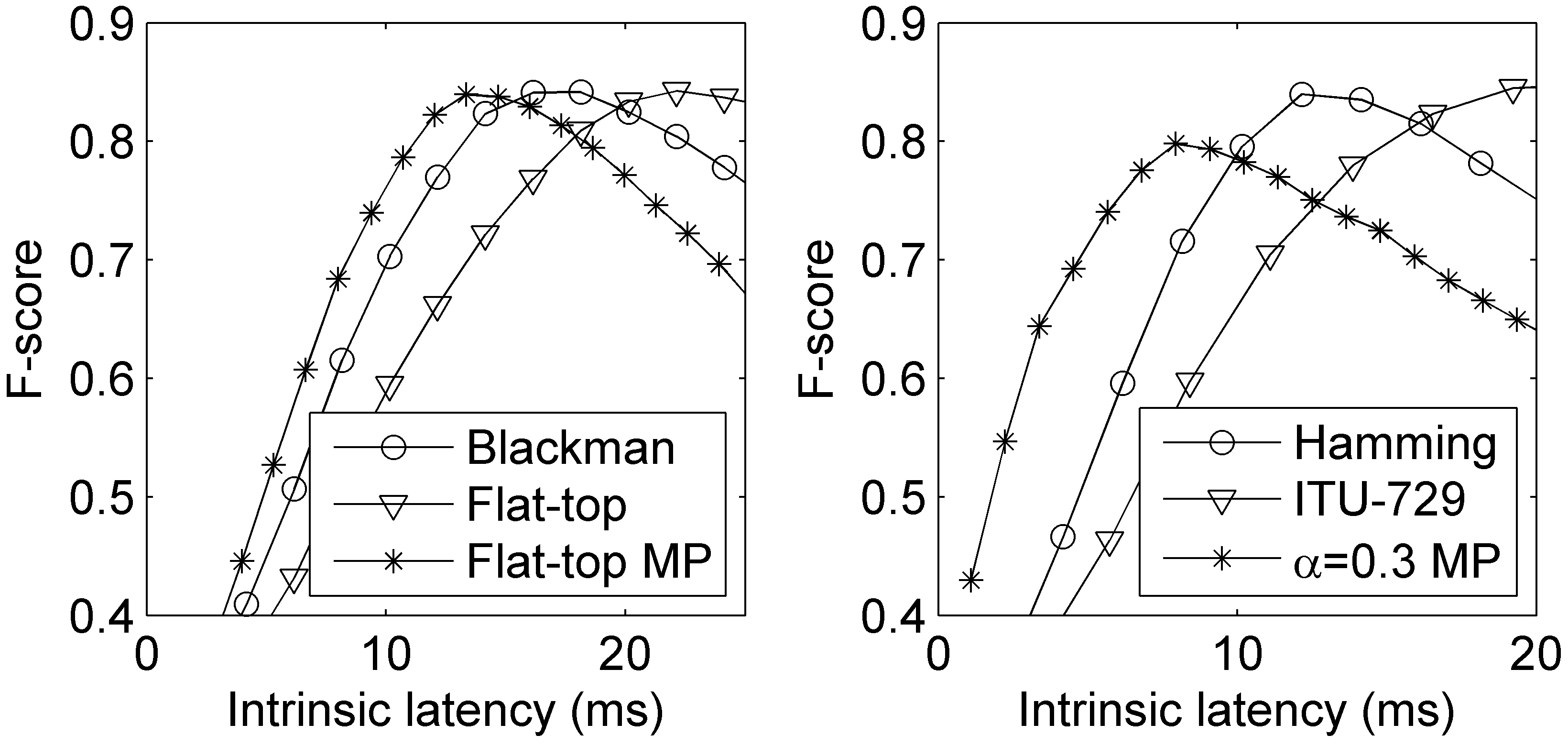}
\caption{Average F-scores of the MAPS dataset using various types of window functions with various intrinsic latencies. Left: Blackman ($h_B$), flat-top ($h_F$), and MP flat-top ($h_{F,\MP}$) windows. Right: Hamming ($h_{0.54}$), $\alpha=0.3$ ($h_{0.30}$), MP $\alpha=0.3$ ($h_{0.30,\MP}$), and ITU-T G.729 ($h_I$) windows.}
\label{fig:maps_stft}
\end{figure}

\subsection{Experiment settings}

All signals are downsampled to 5,512.5 Hz. For all TFRs, we use the same hop size of 2.9 ms (16 samples) and varying window sizes from 5 to 100 ms. When computing FFT we use zero-padding such that the frequency resolution in the TFR is 10 Hz. We consider the following window functions:

\begin{itemize}
\item 3-term cosine series windows: the Blackman ($h_B$), flat-top ($h_F$), and MP flat-top ($h_{F,\MP}$) windows.
\item 2-term cosine series windows: the Hamming ($h_{0.54}$), $\alpha=0.3$ ($h_{0.30}$), and MP $\alpha=0.3$ ($h_{0.30,\MP}$) windows.
\item The ITU-T G.729 ($h_I$) window.
\end{itemize}

We are interested in knowing which window achieves the highest F-score among all window sizes and, more importantly, which window function performs better when the desired intrinsic latency is extremely low (e.g., $t^{h}_{l}<10$ ms).

\subsection{Experiment result of minimal latency windows}

The left diagram of Fig. \ref{fig:maps_stft} shows the result of the 3-term cosine windows ($h_B$, $h_F$, and $h_{F,\MP}$). The three window functions have similar optimal F-scores, but their corresponding intrinsic latencies are quite different. First, by comparing $h_B$ and $h_F$, we found that $h_B$ outperforms $h_F$ for $t^{h}_{l}<20$ ms. More specifically, $h_B$ achieves the optimal F-score (84.15\%) at $t^{h}_{l}=17.96$ ms, while $h_F$ achieves its optimum (84.25\%) at a longer latency $t^{h}_{l}=21.95$ ms. Since the two windows have the same intrinsic latency for the same $N$, such a difference should be caused by the magnitude spectra. In fact, $h_B$ has less spectral leakage than $h_F$ due to its better side-lobe rejection, so $h_F$ requires longer size in order to reduce its spectral leakage effect. It is to say, $h_B$ benefits from its magnitude spectrum rather than its intrinsic latency.

However, when considering the MP window, we found that $h_{F,\MP}$ outperforms $h_B$ for $t^{h}_{l}<15$ ms by taking the advantage of its low intrinsic latency. At $t^{h}_{l}=13.18$ ms, $h_{F,\MP}$ achieves its optimum of 83.76\%, slightly lower than the optimum of $h_F$. In comparison to $h_F$, the intrinsic latency corresponding to the optimal F-score of $h_{F,\MP}$ is reduced by almost 40\% (i.e., from 21.95 ms to 13.18 ms). In brief, there is a ``time-compressing'' behavior for the MP window: for the same window size, the intrinsic latency of the MP window is reduced while the F-score varies little.

The right diagram of Fig. \ref{fig:maps_stft} shows the result of the 2-term cosine windows ($h_{0.54}$, $h_{0.30}$, and $h_{0.30,\MP}$) and the ITU-T G.729 window ($h_I$).
As expected, the window with smaller intrinsic latency performs better than the others when the intrinsic latency is low: $h_I$ performs better than $h_{0.54}$ for $t^{h}_{l}<12$ ms, and $h_{0.30,\MP}$, with smallest intrinsic latency of $0.28N$, further outperforms $h_I$ for $t^{h}_{l}<5.5$ ms.
%{Note that} the intrinsic latency of $h_{0.30,\MP}$ is optimized to its minimal latency $0.28N$, smaller than the one of $h_I$.
However, the globally optimal F-score of $h_{0.30,\MP}$ is only 79.8\% due to its high side-lobe level and the accompanying spectral leakage.

We summarize that an asymmetric MP window would outperform other symmetric or non-MP windows when the desired intrinsic latency of the system is low. In contrast, when looking at a longer latency, the effect of the spectrum should be taken into account. The MP transform provides a way of designing minimal latency window from the symmetric window while keeping spectral information and the system performance.

\begin{figure}
\centering
\includegraphics[width=\columnwidth]{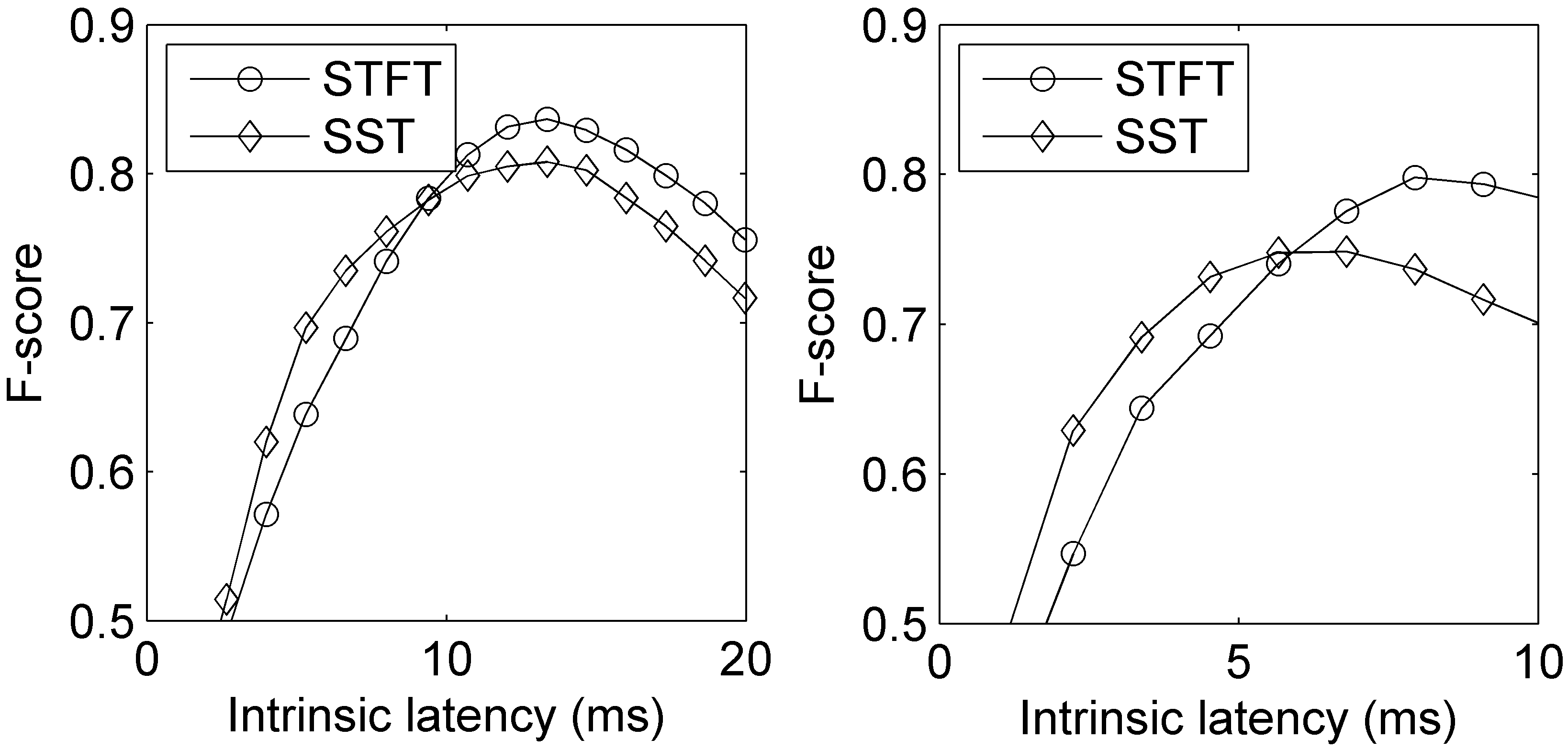}
\caption{Comparison of F-scores using different features (STFT and SST). Left: $V^{(h_{F,\MP})}_x$ and $S^{(h_{F,\MP})}_x$. Right: $V^{(h_{0.30,\MP})}_x$ and $S^{(h_{0.30,\MP})}_x$.}
\label{fig:sst_comparison}
\end{figure}

\subsection{Experiment result of features}
\label{sec:exp_result_of_feature}
The left and the right diagrams of Fig. \ref{fig:sst_comparison} compares the STFT and SST features, $V^{(h_{F,\MP})}_x$ and $S^{(h_{F,\MP})}_x$, and $V^{(h_{0.30,\MP})}_x$ and $S^{(h_{0.30,\MP})}_x$, respectively. In general, STFT outperforms SST when viewing the globally optimal F-score. Although SST enhances the oscillation terms, it is likely to introduce artifacts in the ODF due to its nature of nonlinearity, supposedly a drawback in onset detection task.

However, when the desired intrinsic latency of the system is extremely low, SST outperforms STFT for both windows. In particular, $S^{(h_{0.30,\MP})}_x$ outperforms $V^{(h_{0.30,\MP})}_x$ for $t^{h}_{l}<5.6$ ms, and $S^{(h_{F,\MP})}_x$ outperforms $V^{(h_{F,\MP})}_x$ for $t^{h}_{l}<9.4$ ms. To explain why, notice that when using a short window, STFT suffers from severe spectral leakage. When using SST, the spectral leakage could be accurately reassigned if the frequency bins are sufficiently dense. Therefore, when the window size is small, SST extracts more useful feature than STFT.

In real-world applications, the total latency is the sum of intrinsic latency and computational latency, where the latter depends on the algorithm and hardware design. From the experiment we see the potential of achieving a very low latency onset detection by the proposed scheme, which could be very useful for real-world applications.

\section{Conclusion}

To reduce the intrinsic latency in the TF analysis inherited from the window function, we explore the possibility of utilizing the asymmetric window. A rigorous definition of intrinsic latency for the TF analysis is provided and we theoretically show that TFR determined by SST does have a smaller intrinsic latency.
Further, a systematic method is proposed to construct an asymmetric window from a symmetric one based on the concept of minimum-phase, if the window satisfies some weak conditions.
To illustrate the proposed approach, we show its performance on the music onset detection problem.

\section{Acknowledgement}
Hau-tieng Wu's work is partially supported by Sloan Research Fellow FR-2015-65363.

\bibliographystyle{IEEEtran}
\bibliography{TFanalysis}

\appendix

\section{Stabilize SST by the rejection scheme}\label{sec:SSTvariation}

While SST has been widely applied in different fields (see \cite{Daubechies_Wang_Wu:2016} for a review), like RM, the TFR determined by SST tends to have a ``higher variance'' or ``more speckles'' \cite{Xiao_Flandrin:2007,Lin_Wu_Tsao_Yien_Hseu:2014,Daubechies_Wang_Wu:2016}. This situation is further worsened when the window is asymmetric, since the asymmetric window introduces nonlinear phase deformation.
To simultaneously achieve the low intrinsic latency, preserve the sharpness as well as stabilize the TFR, we consider the following {\em rejection scheme}.

Although theoretically SST could provide a sharp TFR, numerically $\Omega^{(h)}_x$ could be unstable, and hence lead to the instability. To stabilize the SST, we could take (\ref{Expansion:OmegaFunctionOverAdaptive}) into account. Indeed, (\ref{Expansion:OmegaFunctionOverAdaptive}) says that if the STFT contains the correct information about the IF, then the deviation of $\Omega_x^{(h)}(t,\omega)$ from $\omega$ should be bounded by $\Delta$; otherwise it is not guaranteed. Motivated by this fact, we could thus design the following {\em rejection rule}, in addition to the reassignment rule at time $t$
\begin{equation}
\mathfrak{R}_t^{(h)}:=\{\eta;\,|\Omega^{(h)}_x(t,\eta)-\eta|\notin \mathfrak{S}_h\},
\end{equation}
where $\mathfrak{S}_h$ is determined by the support region of $\hat{h}$. With the rejection rule, the SST defined in (\ref{definition:SST}), $S^{(h)}_{x}(t,\xi)$, could be modified to
\begin{align}\label{definition:ModifiedSST}
\int_{\mathfrak{N}_t\backslash \mathfrak{R}_t^{(h)}}V^{(h)}_x(t,\eta)g_\alpha\left(|\xi-\Omega^{(h)}_x(t,\eta)|\right)\ud \eta{\color{blue}\,;}
\end{align}
that is, we replace $\mathfrak{N}_t$ in (\ref{definition:SST}) by $\mathfrak{N}_t\backslash \mathfrak{R}_t^{(h)}$.
We could further consider different window functions to define more rejection rules. One particular approach is by taking the multi-taper idea \cite{Percival:1993} into account. Take $K$ orthogonal windows $h_k$, $k=1,\ldots,K$, which is called the {\em $K$-taper windows}, and we could define $K$ rejection rules, denoted as $\mathfrak{R}_t^{(h_k)}$. Then further replace $\mathfrak{N}_t\backslash \mathfrak{R}_t^{(h)}$ in (\ref{definition:ModifiedSST}) by $\mathfrak{N}_t\backslash \cup_{k=1}^K\mathfrak{R}_t^{(h_k)}$.
In addition to the rejection rules, by all means we can also compute the multi-taper TFR from these windows. However, to simplify the discussion, we do not discuss this direction in this paper; the interested reader is referred to \cite{Daubechies_Wang_Wu:2016} for further information.
In the following, we use the notation $S^{(h_1,h_2,\cdots,h_m)\backslash(g_1,g_2,\cdots,g_n)}_{x}$ to denote the multi-tapered SST computed by $m$ orthogonal windows $h_1,h_2,\cdots,h_m$ with the rejection rule determined by $n$ windows $g_1,g_2,\cdots,g_n$, which can be also chosen from $h_1,h_2,\cdots,h_m$. For example, the traditional SST using one window $h$ and setting the rejection rule determined by that window is denoted as $S^{(h)\backslash(h)}_x$.

We mention that while this approach is proposed to stabilize the TFR determined by SST with an asymmetric window, the original SST based on symmetric windows could also benefit from this rejection scheme. Note that the same idea could be applied to other window-based TF analyses, like the RM or the variations of SST.

\subsection{Example: Respiratory signal}

In the past decades, more and more clinical researches focus on extracting possible hidden dynamics, which are not easily observed by our naked eyes, from the respiratory signal. SST and several other TF analysis methods have proved useful for this purpose to help clinicians' decision making in several problems like the ventilator weaning \cite{Wu_Hseu_Bien_Kou_Daubechies:2013}, sleep apnea treatment, etc. Clearly, to push the application to the monitoring system, a low latency TF analysis tool is necessary. As an important signal, in this subsection, we show how SST works on it with the MP window and/or the rejection rules.

Fig. \ref{fig:respiration} compares several TFRs including $V^{(h)}_x$, $S^{(h)}_x$, $S^{(h)\backslash(h^\prime)}_x$, $S^{(h)\backslash(h,h^\prime)}_x$ and RM on a respiratory signal sampled at 100 Hz. We use the following 2-taper orthogonal windows $\left[h,h^\prime\right]$, where $h$ is the flat-top window and $h^\prime$ is its derivative:
\begin{align}
\label{eq:illustration_h} h(n) &= 0.28 c_0(n) - 0.52 c_1(n) + 0.20 c_2(n), \\
 \label{eq:illustration_g} h^\prime(n) &= - 0.52 s_1(n) + 0.20 s_2(n),
\end{align}
where $s_k(n) = \sin(2\pi kn/T)$. The MP transform of them, namely $h_{\MP}=\mathcal{M}h$ and $h^\prime_{\MP}=\mathcal{M}h^\prime$, respectively, are the corresponding asymmetric windows. The rejection rules are designed to only allow the components supported by the main-lobe regions of $h$ and $h^\prime$.
More specifically, we have
\begin{align}
\label{eq:rejection_h} \mathfrak{R}_t^{(h)}&:=\{\eta;\,|\Omega^{(h)}_x(t,\eta)-\eta| > \Delta\}, \\
\label{eq:rejection_g} \mathfrak{R}_t^{(h')}&:=\{\eta;\,||\Omega^{(h')}_x(t,\eta)-\eta|-\Delta| > \Delta/2\}
\end{align}
and the same for $h_{\MP}$ and $h^\prime_{\MP}$. We set $\Delta=1.5$ bins, a quarter of the main-lobe width of $h$. The window size is 20 seconds and the hop size is 0.2 seconds. The TFRs are illustrated in terms of the estimation time ($t^{(h)}_e$). We clearly observe that the TFRs using the symmetric and asymmetric window look quite similar for the main components, except for a ``time shift'' in the STFT and the three SSTs. The profiles of these TFRs using the MP window appear earlier than the one using the symmetric window by around 5 seconds; this time shift can be estimated by the intrinsic latency of the MP flat-top window (see Table \ref{tab: comp}). It is also clear that RM using the asymmetric window has no such time shift effect, which reflects the theoretical discussion above.

Through a careful examination, we could see more speckle terms and inter-component coupling in the TFR using MP window than the one using symmetric window. This is mainly caused by the nonlinear phase behavior in the MP window, which leads to a discrepancy in the frequency reassignment rules. However, comparing $S^{(h_{\MP})}_x$ to $S^{(h_{\MP})\backslash(h_{\MP})}_x$ and $S^{(h_{\MP})\backslash(h_{\MP},h^\prime_{\MP})}_x$, we found that such discrepancy can be well eliminated by using single- or multi-taper windows for rejection, and the IF curves are made clear enough as in the ones in the TFR using the symmetric window. In conclusion, multi-window rejection approach is an efficient way to obtain a sharply localized TFR, and we would expect more applications of this technique in other signals.

\subsection{Example: Piano signal}

Similarly, Fig. \ref{fig:piano_tfr} further compares the TFRs of a short excerpt of a jazz piano recording, using the same windows (\ref{eq:illustration_h})(\ref{eq:illustration_g}), rejection schemes (\ref{eq:rejection_h})(\ref{eq:rejection_g}), and the asymmetric windows constructed by the MP transform. The sampling rate is 5,512.5 Hz, the window size is 257 and the hop size is 16. Similarly, by using the MP windows each note event appears earlier by around 10 ms than using symmetric windows for all feature except RM. This can also be seen more clearly from the spectral flux features computed from the above-mentioned TFRs, as shown in the bottom of Fig. \ref{fig:piano_tfr}. For each type of TFR, the spectral flux using $h_{\MP}$ (red line) appears earlier than the one using $h$ by around 10 ms also. Using the rejection schemes again suppress the speckle terms, especially for those TFRs using MP windows.

However, we should mention that when analyzing the signals like piano, which has both pitched and percussive counterparts, multi-window rejection method should be carefully treated, depending on what kind of information we want to preserve. The short attack times of piano produces wide-spread spectra (i.e., vertical lines in the TFR) at the onset of every note, followed by a relatively stable harmonic pattern (i.e., horizontal lines in the TFR) revealing the behavior of natural vibration modes of the piano string. As illustrated, in $S^{(h_{\MP})\backslash(h_{\MP})}_x$ and $S^{(h_{\MP})\backslash(h_{\MP},h^\prime_{\MP})}_x$, although we successfully reduce the unwanted fluctuations, the vertical lines characterizing the information of note onsets are also eliminated because they are not narrow-band signals. Therefore, if we wish to preserve the wide-band counterparts, such as for the application like onset detection, too much rejection is not suggested. Detailed discussion will be reported in the future work.

\begin{figure*}
\centering
\includegraphics[width=\textwidth]{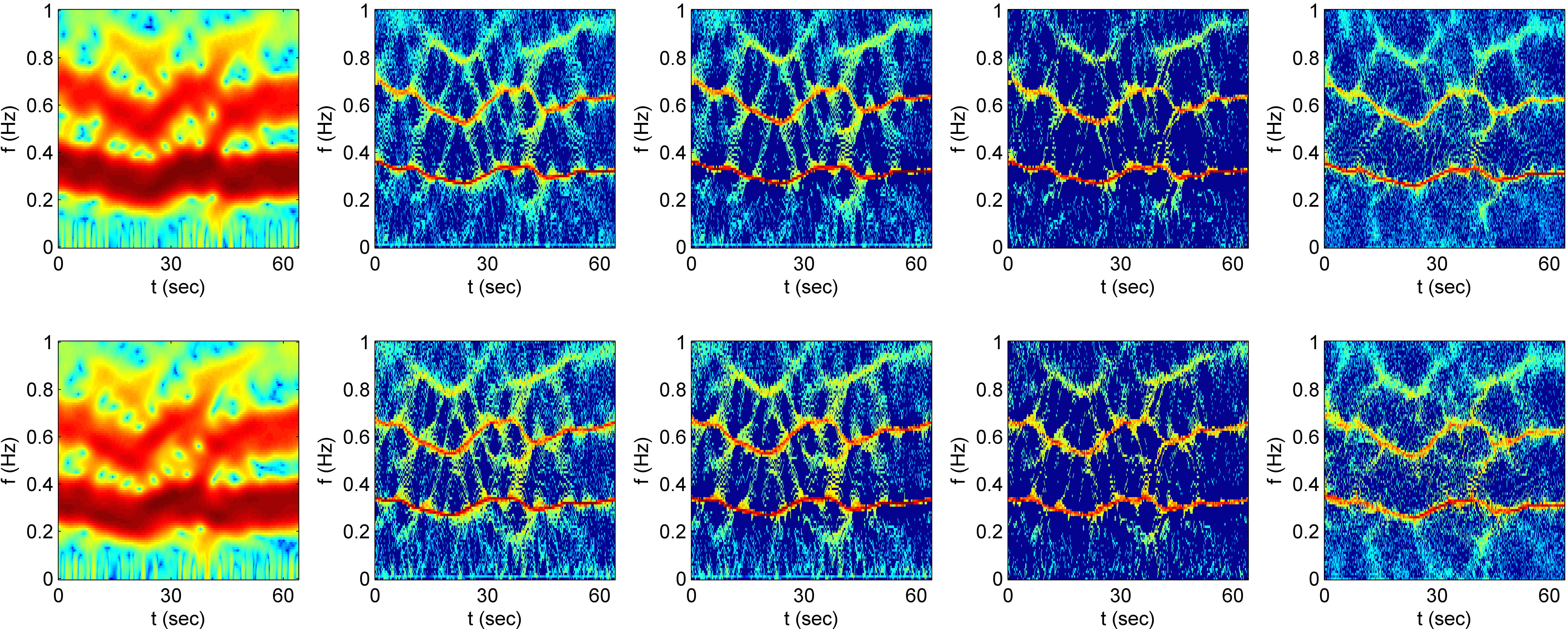}
\caption{Time-frequency representations of a respiratory signal. Upper row, from left to right: $V^{(h)}_x$, $S^{(h)}_x$, $S^{(h)\backslash(h)}_x$, $S^{(h)\backslash(h,h^\prime)}_x$ and $R^{(h)}_x$. Lower row, from left to right: $V^{(h_{\MP})}_x$, $S^{(h_{\MP})}_x$, $S^{(h_{\MP})\backslash(h_{\MP})}_x$, $S^{(h_{\MP})\backslash(h_{\MP},h^\prime_{\MP})}_x$ and $R^{(h_{\MP})}_x$. $h$ is the flat-top window and $h^\prime$ is the first derivative of the flat-top window.}
\label{fig:respiration}
\end{figure*}

\begin{figure*}
\centering
\includegraphics[width=\textwidth]{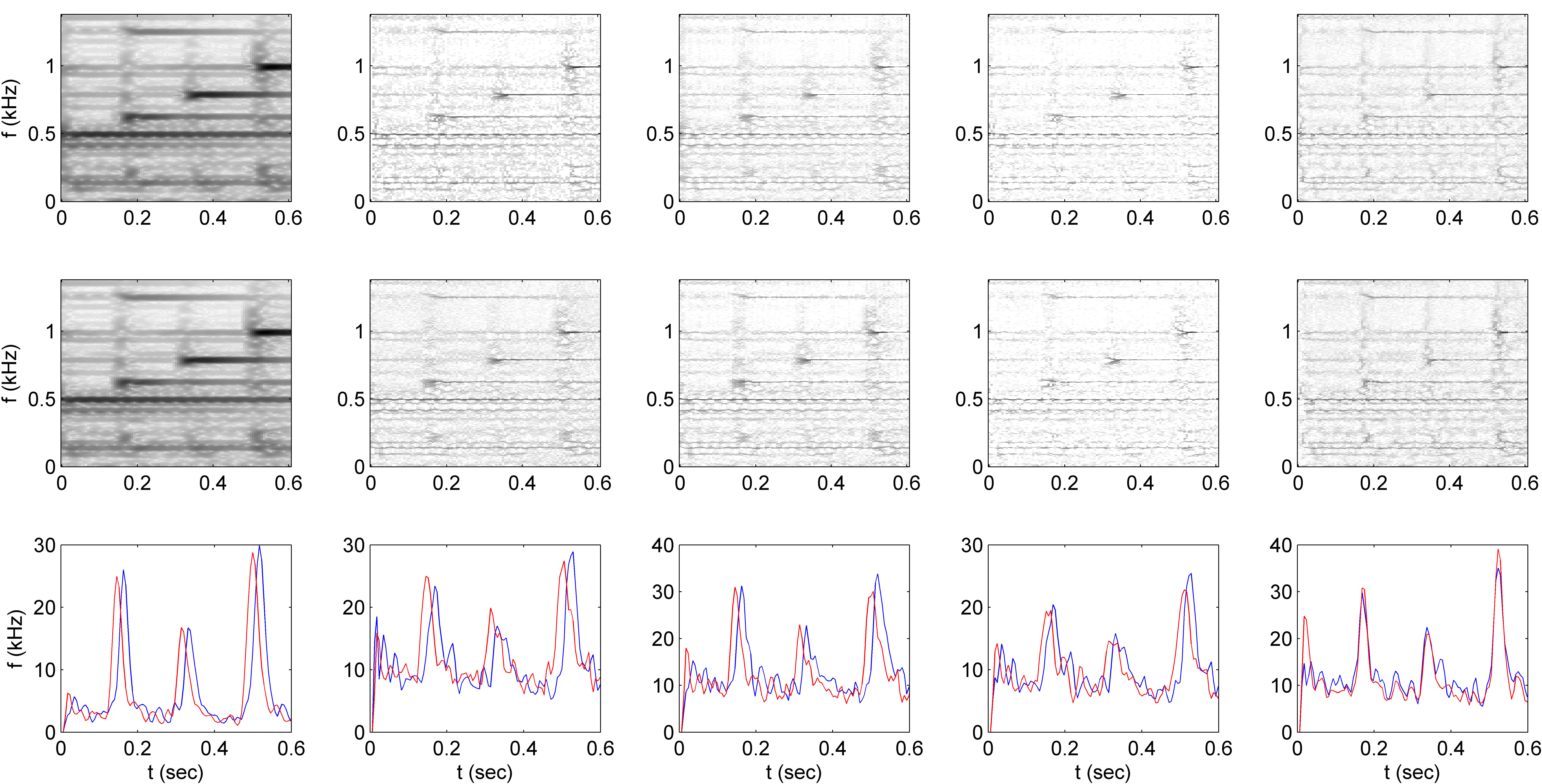}
\caption{Time-frequency representations of a piano signal. Upper row, from left to right: $V^{(h)}_x$, $S^{(h)}_x$, $S^{(h)\backslash(h)}_x$, $S^{(h)\backslash(h,h^\prime)}_x$ and $R^{(h)}_x$. Middle row, from left to right: $V^{(h_{\MP})}_x$, $S^{(h_{\MP})}_x$, $S^{(h_{\MP})\backslash(h_{\MP})}_x$, $S^{(h_{\MP})\backslash(h_{\MP},h^\prime_{\MP})}_x$ and $R^{(h_{\MP})}_x$. Lower row, from left to right: the spectral flux features computed by $V^{(h)}_x$, $S^{(h)}_x$, $S^{(h)\backslash(h)}_x$, $S^{(h)\backslash(h,h^\prime)}_x$ and $R^{(h)}_x$ (blue lines) and $V^{(h_{\MP})}_x$, $S^{(h_{\MP})}_x$, $S^{(h_{\MP})\backslash(h_{\MP})}_x$, $S^{(h_{\MP})\backslash(h_{\MP},h^\prime_{\MP})}_x$ and $R^{(h_{\MP})}_x$ (red lines). $h$ is the flat-top window and $h^\prime$ is the first derivative of the flat-top window.}
\label{fig:piano_tfr}
\end{figure*}

\end{document}